\documentclass[letterpaper, 10 pt,conference]{IEEEconf}

\IEEEoverridecommandlockouts

\makeatletter
\let\NAT@parse\undefined
\def\endthebibliography{%
	\def\@noitemerr{\@latex@warning{Empty `thebibliography' environment}}%
	\endlist
}
\makeatother
\newcommand{\cdconly}[1]{}
\newcommand{\arxivonly}[1]{#1}

\usepackage{graphics} % for pdf, bitmapped graphics files
\usepackage{graphicx}
\usepackage{booktabs}
\usepackage{epsfig} % for postscript graphics files
\usepackage{amsmath} % assumes amsmath package installed
\usepackage{amssymb}  % assumes amsmath package installed
\usepackage{mathrsfs}  
\usepackage{pifont}
\usepackage{xfrac}
\usepackage{tikz} % <--- if tikz pictures are used
\usepackage{pgfplots}
\pgfplotsset{compat=1.16} 
\pgfplotsset{
	/pgfplots/area legend/.style={
		legend image code/.code={\draw [#1] (0.025,0) circle (0.02);
		},
	},
}
\usetikzlibrary{positioning}
\usetikzlibrary{calc}
\usetikzlibrary{decorations.markings,calc,patterns,decorations.pathmorphing}
\usetikzlibrary{arrows, shapes}
\tikzstyle{block} = [draw, thick, node distance=0.5cm, minimum width=1cm, inner sep=6pt]
\tikzstyle{sum} = [draw, thick, circle, node distance=1cm, inner sep=3.5pt, path picture={\draw[thin] (0.12,0) -- (-0.12,0);\draw[thin] (0,0.12) -- (0,-0.12);}]

\usepackage{arydshln}
\usepackage{algpseudocode}
\usepackage{algorithm}

\usepackage{lipsum}
\usepackage{verbatim}

\usepackage{amsthm}

\usepackage{mathtools}
\usepackage{epstopdf}

\usepackage[colorlinks]{hyperref}
\usepackage{color}
\definecolor{myblue}{rgb}{0.00000,0.44700,0.74100}%
\definecolor{myred}{rgb}{0.95000,0.200,0.200}%
\definecolor{mycolor3}{rgb}{0.92900,0.69400,0.12500}%
\definecolor{mygreen}{rgb}{0.1,0.7,0.15}%
\AtBeginDocument{%
\hypersetup{
	colorlinks   = true, %Colours links instead of ugly boxes
	urlcolor     = myblue, %Colour for external hyperlinks
	linkcolor    = myblue, %Colour of internal links
	citecolor    = mygreen %Colour of citations
}
}

\newtheorem{assumption}{\bf Assumption}
\newtheorem{definition}{\bf Definition}
\newtheorem{example}{\bf Example}
\newtheorem{theorem}{\bf Theorem}

\newtheorem{lemma}{\bf Lemma}
\newtheorem{corollary}{\bf Corollary}
\newtheorem{remark}{\bf Remark}

\title{
Output-feedback model predictive control under dynamic uncertainties using integral quadratic constraints
}
\author{Lukas Schwenkel, Johannes K\"ohler, Matthias A. M\"uller, Frank Allg\"ower% <-this % stops a space
\thanks{L. Schwenkel thanks the International Max Planck Research School for Intelligent Systems (IMPRS-IS) for supporting him. L. Schwenkel and F. Allg\"ower are with the University of Stuttgart, Institute for Systems Theory and Automatic Control, Germany.
(email: $\{$schwenkel, allgower$\}$@ist.uni-stuttgart.de). J. K\"ohler is with the ETH Zürich, Institute for Dynamical Systems and
Control, Switzerland. (email: jkoehle@ethz.ch). M.A. M\"uller is with the Leibniz University Hannover, Institute of Automatic Control, Germany. (email: mueller@irt.uni-hannover.de).}
}
\newcommand{\refeqleft}[2]{\overset{\makebox[0pt][c]{\scriptsize #1}}{#2}\hspace{\widthof{\scriptsize #1}/2-\widthof{$#2$}/2}}
\renewcommand{\refeq}[2]{\overset{\makebox[0pt][c]{\scriptsize #1}}{#2}}

\newcommand{\x}{x}
\newcommand{\p}{p}
\newcommand{\w}{w}
\renewcommand{\d}{d}
\newcommand{\dmax}{\gamma_\d}
\renewcommand{\u}{u}
\newcommand{\q}{q}
\newcommand{\z}{z}

\newcommand{\y}{y}
\newcommand{\filty}{s}
\renewcommand{\t}{t}
\renewcommand{\k}{k}
\newcommand{\A}[1]{A_{#1}}
\newcommand{\B}[2]{B_{#1}^{#2}}
\newcommand{\C}[2]{C_{#1}^{#2}}
\newcommand{\D}[3]{D_{#1}^{#2#3}}
\newcommand{\BB}[1]{B_{#1}}
\newcommand{\CC}[1]{C_{#1}}
\newcommand{\DD}[1]{D_{#1}}

\newcommand{\CK}{C_\K}
\newcommand{\DK}{D_\K}
\newcommand{\G}{G}                            % G system

\newcommand{\K}{K}                            % K controller
\newcommand{\Kx}{\kappa}                      % kappa controller state

\renewcommand{\L}{L}                          % L observer
\newcommand{\Lrho}{{L_\rho}}
\newcommand{\GK}{{\Theta}}                   % \Theta interconnection G \star K
\newcommand{\GKrho}{{\GK_\rho}}

\newcommand{\filt}{\Psi}                      % Psi observer 
\newcommand{\filtx}{\psi}                     % psi observer state
\newcommand{\all}{\Sigma}                     % Sigma interconnection of Psi and G \star K
\newcommand{\allrho}{{\Sigma_\rho}}

\newcommand{\obs}{\Xi}                        % Xi interconnection of Psi, GK, and L
\newcommand{\obsrho}{{\Xi_\rho}}
\newcommand{\gpobs}{\all}                     % Phi defined by Phi \star L = Xi
\newcommand{\gpobsrho}{{\allrho}}

\newcommand{\Pobs}{\P^{\mathrm{o}}}
\newcommand{\Psynobs}{\Psyn^{\mathrm{o}}}

\newcommand{\Ysyntrafoobs}{{\Ysyntrafo^{\mathrm{o}}}}
\newcommand{\Zsynobs}{\mathcal{Z}^{\mathrm{o}}}

\newcommand{\Mobs}{\M^{\mathrm{o}}}
\newcommand{\muobs}{\mu^{\mathrm{o}}}
\newcommand{\gamobs}{\gamma^{\mathrm{o}}}
\newcommand\Psyn{\mathcal{\P}}

\newcommand\Ksynobs{\mathcal{K}^\mathrm{o}}
\newcommand\Lsynobs{\mathcal{L}^\mathrm{o}}
\newcommand\Msynobs{\mathcal{M}^\mathrm{o}}
\newcommand\Nsynobs{\mathcal{N}^\mathrm{o}}

\newcommand\Asynobs{\mathcal{A}^\mathrm{o}}
\newcommand\Bsynobs{\mathcal{B}^{\mathrm{o}}}
\newcommand\Csynobs{\mathcal{C}^{\mathrm{o}}}
\newcommand\Dsynobs{\mathcal{D}^{\mathrm{o}}}

\newcommand\Ysyntrafo{\mathcal{Y}}
\DeclareMathOperator{\diag}{diag}

\newcommand{\symb}[1]{\begin{pmatrix}\vphantom{#1}\star \end{pmatrix}^{\hspace{-0.15cm}\top}\hspace{-0.1cm} #1 }
\newcommand{\symbscalar}{(\star)^{\hspace{-0.05cm}\top}\hspace{-0.025cm} }

\newcommand{\diagmat}[1]{{\arraycolsep=2pt\begin{pmatrix} #1\end{pmatrix}}}
\newcommand{\nx}{{n_\x}}
\let\greeknu\nu
\renewcommand{\nu}{{n_\u}}
\newcommand{\nw}{{n_\w}}
\newcommand{\nd}{{n_\d}}
\newcommand{\nq}{{n_\q}}
\newcommand{\np}{{n_\p}}
\newcommand{\ny}{{n_\y}}
\newcommand{\nz}{{n_\z}}
\newcommand{\nK}{{n_\Kx}}

\newcommand{\nfilty}{{n_\filty}}

\newcommand{\nfilt}{{n_\filtx}}

\newcommand{\M}{M}
\newcommand{\X}{X}
\renewcommand{\P}{P}

\DeclareSymbolFont{newfont}{OML}{cmm}{m}{it}% Computer Modern math font
\newcommand{\R}{\mathbb{R}}
\newcommand{\N}{\mathbb{N}}
\newcommand{\I}{\mathbb{I}}
\newcommand{\Kinf}{\mathscr{K}_\infty}
\newcommand{\KL}{\mathscr{K\!L}}
\renewcommand{\S}{\mathbb{S}}
\newcommand{\Hinf}{\mathcal{H}_\infty}

\newcommand{\lt}{\bar}
\newcommand{\blkmat}[2]{\left(\begin{array}{@{}#1@{}} #2 \end{array}\right)}
\newcommand{\newblkdash}[1][2.25ex]{\\[0.02cm] \hdashline\rule{0pt}{#1}}

\newcommand{\peak}{{\mathrm{peak}}}%\newcommand{\peak}{{\infty}}
\usepackage{stmaryrd}

\newcommand{\MXset}{\mathbb{\M\X}}

\newcommand{\GKx}{\theta}
\newcommand{\GKxe}{\ubar{\GKx}}

\newcommand{\eGKxe}{\delta\GKxe}
\newcommand{\initopt}{\greeknu}
\newcommand{\GKxnp}{\tilde{\GKx}}
\newcommand{\eGKxnp}{\delta\GKxnp}
\newcommand{\eGKxnpe}{\underline{\eGKxnp}}
\newcommand{\ynp}{\tilde{\y}}
\newcommand{\qnp}{\tilde{\q}}

\newcommand{\eynp}{\delta\tilde{{\y}}}
\newcommand{\eqnp}{\delta\tilde{{\q}}}
\newcommand{\xnf}{\hat{\x}}

\newcommand{\GKxnf}{\hat{\GKx}}
\newcommand{\eGKxnf}{\delta\hat{\GKx}}

\newcommand{\qnf}{\hat{\q}}
\newcommand{\znf}{\hat{\z}}
\newcommand{\unf}{\u^{\mathrm{MPC}}}
\newcommand{\exnf}{\delta{\hat{\x}}}

\newcommand{\eznf}{\delta{\hat{\z}}}
\newcommand{\eqnf}{\delta\hat{\q}}

\newcommand{\xe}{\ubar{\x}}
\newcommand{\exe}{{\delta\xe}} % estimation error
\newcommand{\Lx}{\lambda}   % observer state
\newcommand{\nL}{{n_\Lx}}
\newcommand{\zobs}{z^\mathrm{o}}
\newcommand{\wobs}{w^\mathrm{o}}
\newcommand{\barwobs}{\bar{w}^\mathrm{o}}
\newcommand{\barzobs}{\bar{z}^\mathrm{o}}
\newcommand{\cf}{\hat{c}}
\newcommand{\cp}{\tilde{c}}
\newcommand{\ce}{\ubar{c}}
\newcommand{\tcw}{\mathcal{S}}

\newcommand{\tswroot}{{\mathcal{T}}}
\newcommand{\scw}{\mathcal{Q}}

\usepackage{accents}
\newcommand{\ubar}[1]{\underaccent{\bar}{#1}}

\begin{document}

\maketitle

%%%%%%%%%%%%%%%%%%%%%%%%%%%%%%%%%%%%%%%%%%%%%%%%%%%%%%%%%%%%%%%%%%%%%%%%%%%%%%%%
\begin{abstract}
	In this work, we propose an output-feedback tube-based model predictive control (MPC) scheme for linear systems under dynamic uncertainties that are described via integral quadratic constraints (IQC). 
	By leveraging IQCs, a large class of nonlinear and dynamic uncertainties can be addressed. 
	We leverage recent IQC synthesis tools to  design a dynamic controller and an estimator that are robust to these uncertainties and minimize the size of the resulting constraint tightening in the MPC.
	Thereby, we show that the robust estimation problem using IQCs with peak-to-peak performance can be convexified.
	We guarantee recursive feasibility, robust constraint satisfaction, and input-to-state stability of the resulting MPC scheme.
\end{abstract}

\section{Introduction}
Model predictive control (MPC) is a popular control strategy due to its ability to ensure constraint satisfaction~\cite{Rawlings2017}. 
The presence of dynamic uncertainties poses a major challenge in designing robust controllers, especially in ensuring constraint satisfaction, which we address in this paper. 
We utilize integral quadratic constraints (IQCs)~\cite{Megretski1997,Veenman2016,Scherer2022a} to describe a wide range of structured and unstructured uncertainties, e.g., $\ell_2$-gain bounded dynamic uncertainties, uncertain parameters or delays, and sector- or slope-restricted static nonlinearities (cf. IQC libraries in~\cite{Megretski1997,Veenman2016}).
Based on this IQC description, we construct a tube confining all possible trajectories (cf.~\cite{Rawlings2017,Chisci2001,Mayne2001}).
We minimize the size of this tube by designing an output-feedback controller using the peak-to-peak gain minimization procedure from~\cite{Schwenkel2025}.
To initialize the MPC predictions, we design a robust estimator with guaranteed error bounds.
Thereby, we show that the discrete-time robust estimation problem using IQCs for peak-to-peak performance can be reformulated as a single convex semi-definite program (SDP)\cdconly{\footnote{This result is only in the extended version~\cite{Schwenkel2025b} due to space limitations.}}, which is a contribution of independent interest. 
A similar reformulation of the continuous-time robust estimation problem using IQCs for $\Hinf$-performance is known from~\cite{Scherer2008}. 
We prove that the proposed MPC scheme is recursively feasible and that the resulting closed loop is input-to-state stable and robustly satisfies the constraints.
In our numerical example, we show that even in the special case of state measurement, the proposed methodology reduces conservatism compared to existing IQC-based MPC approaches~\cite{Schwenkel2020,Schwenkel2023c} as the IQC-filter states are estimated. 

\emph{Related work.} Dynamic uncertainties in tube-based MPC have been addressed by~\cite{Lovaas2008,Falugi2014}, where the tube is constructed based on a peak-to-peak gain bound on the uncertainty in combination with a rigid uniform bound on the peak of the output signal that enters the uncertainty. 
We avoid the use of a rigid uniform bound and instead capture how the peak of the error depends on the control input.
Thereby, the MPC can decide online whether to tighten or loosen the tube, resulting in significantly larger flexibility and reduced conservatism.
Such a flexible tube approach has been used in~\cite{Loehning2014} as well for the special case where the uncertainty results from using reduced order models.
In contrast to these works, the structure and nature of a variety of uncertainties can be described in a less conservative fashion by using IQCs~\cite{Megretski1997,Veenman2016,Scherer2022a}.
We presented a similar IQC based-approach in~\cite{Schwenkel2020} and extended it in~\cite{Schwenkel2023c} to include measurements in the prediction by an initial state optimization which improved the overall performance.
However, both schemes are limited to state measurements and do not provide a systematic offline design procedure of the stabilizing controller.
In the absence of dynamic uncertainties, classic output-feedback MPC designs use a tube based on the worst-case estimation error~\cite{Mayne2006}. 
In~\cite{Subramanian2017}, it was shown that the conservatism of output-feedback MPC can be reduced by including the knowledge from the previous prediction and error bound into the initial state optimization. 
In a similar spirit, we optimize the initial state by interpolating between the estimate and the previous prediction, as well as the estimation error bound and the previous prediction error bound, similar to the initial state interpolation schemes from~\cite{Schlueter2022} and~\cite{Koehler2022}.

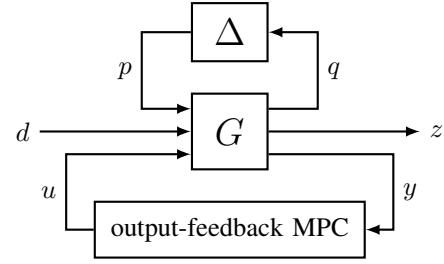
\begin{figure}[tb]\centering
	\begin{tikzpicture}
		\node [block] (Delta) {\Large$\Delta$};
		\draw[white] (Delta.west) + (-0.8,0) -- +(-1.2,0) node[left]{$\p$};
		\node [block,node distance=0.4cm, below=of Delta, minimum height=1cm, minimum width=1cm] (G) {\Large $\G$};
		% d
		\draw[latex-, thick] (G.west) -- +(-2,0)node [left] {$\d$};
		% K
		\node [block,node distance=0.4cm, below = of G] (K) {output-feedback MPC};
		\draw[-latex, thick] (G.east)+(0,-0.3)  --  +(1.65,-0.3) |-node [right,pos=0.25] {$\y$} (K);
		% input
		\draw[latex-, thick] (G.west)+(0,-0.3)  -- +(-1.65,-0.3)|- node [left,pos=0.25] {$\u$} (K);
		\draw[-latex, thick] (G.east) +(0,0) -- ($(G.east)+(2,0)$) node [right] {$\z$};
		\draw[latex-, thick] (Delta.east) -| node[right, pos=0.75] {$\q$} ($(G.east)+(0.65,0.3)$) -- ($(G.east)+(0,0.3)$);
		\draw[-latex, thick] (Delta.west) -| node[left, pos=0.75] {$\p$} ($(G.west)+(-0.65,0.3)$) -- ($(G.west)+(0,0.3)$);
	\end{tikzpicture}
	\caption{Interconnection of system $\G$, uncertainty $\Delta$, and MPC controller.\\[-2cm] }\label{fig:DGK}
\end{figure} 
\emph{Outline.} After introducing the problem setup in Section~\ref{sec:setup}, we perform a robust reachability analysis to construct the tube in Section~\ref{sec:tube}. 
The robust MPC scheme is proposed in Section~\ref{sec:rmpc} and recursive feasibility and stability are proven based on a certain assumption on the estimation error.
In Section~\ref{sec:estimator}, we show how to construct a robust estimator satisfying this assumption. 

\emph{Notation.} Vertically stacked vectors $x\in\R^n$, $y\in\R^m$ are denoted by $[x;\,y]\in\R^{n+m}$. The set of integers in the interval $[a,b]$ is denoted by $\I_{[a,b]}$. For $x\in\R^n$ and $a\in\R$ we write $x\leq a$ if $x_i\leq a$ for all $i\in\I_{[1,n]}$.
The set of symmetric matrices in $\P\in\R^{n\times n}$ with $\P=\P^\top$ is denoted by $\S^n$.	
For $A\in \R^{n\times m}$ and $P \in\R^{i\times j}$, denote $\diag(A,P)\triangleq (\begin{smallmatrix} A\\&P \end{smallmatrix})\triangleq (\begin{smallmatrix} A&0\\0&P \end{smallmatrix})\in\R^{(n+i)\times (m+j)}$ and $\symbscalar{PA}\triangleq A^\top P A$  if $\P\in \S^{n}$.
Further, we use $\star$ ($\bullet$) to denote symmetric (irrelevant) entries in block matrices, e.g., $(\begin{smallmatrix} I & A \\ \star & I\end{smallmatrix})=(\begin{smallmatrix} I & A \\ A^\top & I \end{smallmatrix})$, whereas $(\begin{smallmatrix} I & A \\ \bullet & I\end{smallmatrix})=(\begin{smallmatrix} I & A \\ B & I \end{smallmatrix})$ for some matrix $B$. 
The set of all signals $x:\N\to\R^n$ is denoted by $\ell_{2,\mathrm{e}}^n$. For $x\in\ell_{2,\mathrm{e}}^n$, define $\|x\|_\peak \triangleq \sup_{t\in\N} \|x_t\|_2$.
The class of continuous increasing functions $\alpha:[0,\infty)\to[0,\infty)$ with $\alpha(0)=0$ and $\lim_{x\to\infty}\alpha(x)=0$ is denoted by $\Kinf$. 
The class of functions $\beta:[0,\infty)\times\N \to [0,\infty)$ with $\beta(\cdot,t)\in\Kinf$ for all $t\in\N$, $\beta(x,\cdot)$ non-increasing and $\lim_{t\to\infty}\beta(x,t)=0$ for all $x\in[0,\infty)$ is denoted by $\KL$.

\section{Problem Setup} \label{sec:setup}
We consider the problem of designing an MPC controller for the interconnection of a known linear system $\G$ 
\begin{subequations}\label{eq:sys}
	\begin{align}
		\x_{t+1} &= \A\G \x_t + \B\G\p   \p_t + \B\G\d \d_t + \B\G\u \u_t \\
		\q_t &= \C\G\q \x_t  + \D\G\q\p \p_t + \D\G\q\d \d_t + \D\G\q\u \u_t\\
		\z_t &= \C\G\z \x_t  + \D\G\z\p \p_t + \D\G\z\d \d_t +  \D\G\z\u \u_t\\
		\y_t &= \C\G\y \x_t  + \D\G\y\p \p_t + \D\G\y\d \d_t
	\end{align}
	and an (possibly dynamic and nonlinear) uncertainty $\Delta$ 
	\begin{align}
		\p_t &= (\Delta(\q))_t
	\end{align}
\end{subequations}
as shown in Figure~\ref{fig:DGK}.
The time index is $t\in\N$, the state $\x_t\in\R^{\nx}$, the control input $\u_t\in\R^{\nu}$, and the uncertainty channel is described by $\p_t\in\R^{\np}$ and $\q_t\in\R^{\nq}$.
The system is subject to polytopic constraints 
\begin{align}\label{eq:constraints}
	\z_t \leq 1, \qquad \forall t\in\N,
\end{align}
where $\z_t\in\R^{\nz}$.
At time $t=0$ an initial estimate $\xnf_0$ for $\x_0$ is available.
Only the output $\y_t\in\R^{\ny}$ can be measured.
The signal $\d_t\in\R^{\nd}$ contains external disturbances and measurement noise, and while it is unknown, we assume that we know a peak bound $\|\d\|_\peak\leq \dmax$.
To describe the dynamic uncertainty $\p=\Delta(\q)$, we use finite horizon IQCs with a terminal cost (cf.~\cite{Scherer2018,Scherer2022a}) and the loop transformation (cf.~\cite{Hu2016}) defined by $T_{\rho^{-1}}: (\q_t)_{t\in\N} \mapsto (\rho^{-t}\q_t)_{t\in\N}$ for $\rho\in(0,1)$.
In particular, let $\lt \q \triangleq T_{\rho^{-1}}\q$, $\lt \p\triangleq T_{\rho^{-1}}\p$, $\Delta_\rho = T_{\rho^{-1}}\circ \Delta \circ T_{\rho}$, i.e., $\lt \p = \Delta_\rho(\lt\q)$. 
\begin{definition}\label{def:iqc}
	An operator $\Delta_\rho:\ell_{2,\mathrm{e}}^{\nq}\to\ell_{2,\mathrm{e}}^{\np}$ satisfies the finite horizon IQC with terminal cost defined by $(\A\filt,\B\filt\q,\B\filt\p,\C\filt\filty,\D\filt\filty\q,\D\filt\filty\p,\X,\M)$ iff for all $\lt\q \in \ell_{2,\mathrm{e}}^\nq$, $\lt\p = \Delta_\rho(\lt\q)$, and $t\in\N$ it holds that
	\begin{align}\label{eq:iqc}
		\sum_{k=0}^{t-1} \filty_{k}^\top \M \filty_k + \filtx_t^\top \X\filtx_t \geq 0
	\end{align}
	with $\filty_t\in\R^{\nfilty}$ and $\filtx_t\in\R^{\nfilt}$ being defined by $\filtx_0=0$ and
	\begin{subequations}\label{eq:filt_dyn}
		\begin{align}
			\filtx_{t+1} &= \A\filt \filtx_t + \B\filt\q \lt \q_t   + \B\filt\p \lt \p_t \\
			\filty_t &= \C\filt\filty \filtx_t + \D\filt\filty\q \lt \q_t + \D\filt\filty\p \lt \p_t.
		\end{align}
	\end{subequations}
\end{definition}
Definition~\ref{def:iqc} characterizes an uncertain operator $\Delta_\rho$ in terms of a known dynamical filter $\filt$, a known multiplier $\M$, and a known terminal cost $\X$ (cf.~\cite{Megretski1997,Veenman2016} and \cite{Scherer2022a}).
\begin{assumption}\label{ass:iqc}
	There exist $\rho\in(0,1)$, a set $\MXset\subseteq \mathbb{S}^{\nfilty}\times \mathbb{S}^{\nfilt}$, and a filter $(\A\filt,\B\filt\q,\B\filt\p,\C\filt\filty,\D\filt\filty\q,\D\filt\filty\p)$ such that $\Delta_\rho$ satisfies the finite horizon IQC with terminal cost defined by $(\A\filt,\B\filt\q,\B\filt\p,\C\filt\filty,\D\filt\filty\q,\D\filt\filty\p,\X,\M)$ for all $(\M,\X)\in\MXset$.
\end{assumption}
The control goal is to guarantee constraint satisfaction~\eqref{eq:constraints} and input-to-state stability from the disturbance input $\d$ to the state $\x$ for all possible disturbances $\|\d\|_\peak \leq \dmax$ and all uncertainties $\Delta$ satisfying Assumption~\ref{ass:iqc}.
We approach this problem by a tube-based MPC scheme, i.e., we confine all possible system trajectories in a sequence of sets called the tube, which we use in the MPC predictions for a robust planning to ensure constraint satisfaction. 

%%%%%%%%%%%%%%%%%%%%%%%%%%%%%%%%%%%%%%%%%%%%%%%%%%%%%%%%%%%%%%%%%%%%%%%%%%%%%%%%%%%%%%%%%%%%%%%%%%%%%%%%%%%%%%%%%%%%%%%%%%%%%%%%%%%%%%%%%%%%%%%%%%%%%%%%%%
\section{Tube construction}\label{sec:tube}
In this section, we construct a tube that we can use for robust planning and we explain how we can minimize its size to reduce conservatism. 
As standard (cf.~\cite{Rawlings2017,Mayne2001,Chisci2001}), we include a stabilizing feedback $\K$ in the robust plan to ensure that the tube remains bounded.
Due to the output-feedback setting, we use a dynamic $\K$ of the form
\begin{subequations}\label{eq:prestab}
	\begin{align}\label{eq:Kx_sys}
		\Kx_{t+1} & = \A\K \Kx_{t} + \BB\K \y_{t}   \\
		\u_t &= \CC\K \Kx_{t} + \DD\K \y_{t}+\unf_t \label{eq:mpc_input}
	\end{align}	
\end{subequations}
where $\u_t$ is the input that is actually applied to the system~\eqref{eq:sys} and $\unf_t$ is computed by the MPC optimization problem. 
The controller $\K$ is initialized with $\Kx_0=0\in\R^{\nK}$.
The MPC scheme optimizes at each time $t$ a sequence of inputs $\unf_{k|t}$ for the next $N$ time points $k\in\I_{[t,t+N-1]}$.
The dynamics of the combined state $\GKx \triangleq [\x;\, \Kx]$ with system~\eqref{eq:sys}, controller~\eqref{eq:prestab}, and an input sequence $\unf_{k|t}$ can be compactly written as 
\begin{subequations}\label{eq:DGKx_sys}
	\begin{align}\label{eq:GKx_sys}
		\left[\GKx_{k+1|t};\, \q_{k|t};\, \z_{k|t};\, \y_{k|t}\right] &= \Theta \left[
			\GKx_{k|t};\, \p_{k|t};\, \d_k;\, \unf_{k|t} \right] \\
		\p_{k|t} &= (\Delta(\q_{\cdot|t}))_k \label{eq:pf}
	\end{align}	
\end{subequations}
where the matrix $\Theta$ is given in~\eqref{eq:GK} in the appendix and where $\GKx_{t|t} = \GKx_t$.
To initialize the dynamic uncertainty correctly, we set $\q_{k|t} \triangleq \q_k$ for $\k\in \I_{[0,t-1]}$.
In MPC, only the first input $\unf_t=\unf_{t|t}$ is applied. 
Since the output $\p_{k|t}$ of the uncertainty $\Delta$ and the disturbances $\d_k$ are unknown, we use the following nominal prediction model 
\begin{align}\label{eq:xnf_sys}
	\left[\GKxnf_{k+1|t};\, \qnf_{k|t};\, \znf_{k|t};\, \bullet \right] &=
	\Theta \left[\GKxnf_{k|t};\, 0;\, 0;\, \unf_{k|t}\right].
\end{align}
The prediction dynamics are free of uncertainties and hence, $\GKxnf_{k|t}$, $\qnf_{k|t}$, and $\znf_{k|t}$ can be computed at time $t$ for all $k\geq t$.
Define the error $\eGKxnf\triangleq \GKx-\GKxnf$, $\eqnf\triangleq \q-\qnf$, $\eznf\triangleq \z-\znf$, which satisfies
\begin{align}\label{eq:exnf_sys}
	\left[\eGKxnf_{k+1|t};\, \eqnf_{k|t};\, \eznf_{k|t};\, \bullet\right] &= \Theta \left[
		\eGKxnf_{k|t};\, \p_{k|t};\, \d_t;\, 0 \right]
\end{align}
where $\p_{k|t}$ follows~\eqref{eq:pf}.
Further, we denote the filter state $\filtx_{k|t}$ and output $\filty_{k|t}$ which result from~\eqref{eq:filt_dyn} initialized at $\filtx_{0|t}\triangleq 0$ and $\lt \p_{k|t}\triangleq \rho^{-k}\p_{k|t}$, $\lt \q_{k|t}\triangleq \rho^{-k}\q_{k|t}$, where $\p_{k|t} \triangleq \p_k$ for $\k\in \I_{[0,t-1]}$.
To satisfy~\eqref{eq:constraints} by using tightened constraints on the nominal prediction $\znf_{k|t}$, we need a bound on $\eznf_{k|t}$. 
For the analysis of~\eqref{eq:exnf_sys}, let $\w_{k|t} \triangleq \begin{bmatrix} \qnf_{k|t};\, \d_{k} \end{bmatrix}$, $\lt\w_{k|t} \triangleq \rho^{-k} \w_{k|t}$, $\overline{\eGKxnf}_{k|t}\triangleq \rho^{-k}  {\eGKxnf}_{k|t}$, and $\overline{\eznf}_{k|t}\triangleq \rho^{-k}  {\eznf}_{k|t}$, then the augmented system dynamics with state $\chi_{k|t} \triangleq \begin{bmatrix} \filtx_{k|t};\, \overline{\eGKxnf}_{k|t} \end{bmatrix}$ for $k\geq t$ is given by
\begin{subequations}\label{eq:aug_sys}
	\begin{align}
		\chi_{\k+1|t} &= \A\allrho \chi_{k|t} + \B\allrho\p \lt\p_{k|t} + \B\allrho\w \lt \w_{k|t} \\
		\filty_{k|t} &= \C\all\filty \chi_{k|t} + \D\all{\filty}\p \lt\p_{k|t} + \D\all{\filty}\w \lt\w_{k|t}\\
		\overline{\eznf}_{k|t} &= \C\all\z \chi_{k|t}  + \D\all\z\p \lt\p_{k|t} + \D\all\z\w \lt\w_{k|t},
	\end{align}
\end{subequations}
where all matrices are defined in~\eqref{eq:ABCDgpobs} in the appendix.
Additionally, define $\exnf_0\triangleq \x_0-x_0$, $\GKxnf_{0|-1}\triangleq[\xnf_0;\,0]$, $\eGKxnf_{0|-1}=[\exnf_0;\,0]$, and $\chi_{0|-1} = [0;\,\eGKxnf_{0|-1}]$.
We assume that the controller $\K$ satisfies the following matrix inequalities.
\begin{assumption}\label{ass:p2p_LMIs}
	There exist $(\M_1,\X_1)\in\MXset$, $(\M_2,\X_2)\in\MXset$, $\P\in\mathbb{S}^{\nx+\nK+\nfilt}$, and $\gamma\geq \mu\geq 0$ such that
	\begin{align}\label{eq:stab_LMI1}
		&\!\!\symb {\diagmat{-\P\\&\P\\&&\M\\&&&- \mu I}\hspace{-0.1cm}
			\begin{pmatrix}
				I &  0 & 0\\
				\A\allrho & \B\allrho\p & \B\allrho\w \\[0.5mm]
				\C\all\filty & \D\all{\filty}\p & \D\all{\filty}\w \\
				0& 0&I
		\end{pmatrix}}\prec 0\\
			&\!\!\resizebox{0.9\linewidth}{!}{$\symb{{\arraycolsep=1pt\begin{pmatrix}\ubar\X_1-\P\\
						&\ubar \X_2 \\
						&&\M_2\\
						&&&\frac {\alpha}{\gamma} I\\
						&&&&-\beta I\end{pmatrix}}\hspace{-0.1cm}
				\begin{pmatrix}I&0&0\\[0.3mm]
					\A\allrho & \B\allrho\p & \B\allrho\w \\[0.8mm]
					\C\all\filty&\D\all\filty\p&\D\all\filty\w\\[0.8mm]
					\C\all\z & \D\all\z\p & \D\all\z\w \\[0.5mm]
					0&0&I\end{pmatrix}}$}  \prec 0\nonumber\\[-0.5cm]
					\label{eq:ana_LMI2} \\
					\label{eq:stab_LMI2}
					&P-  \ubar\X_1-\ubar\X_2\succ 0,
	\end{align}
	where $\ubar \X_i = \diag(\X_i,0)$ for $i\in\{1,2\}$, $\M=\M_1+\M_2$, $\alpha = \frac{\rho^2}{1-\rho^2}$, and $\beta = \alpha (\gamma-\mu)$.	
\end{assumption}
In~\cite[Theorem~3]{Schwenkel2025} it is shown that~\eqref{eq:stab_LMI1},~\eqref{eq:ana_LMI2}, and~\eqref{eq:stab_LMI2} imply that the peak-to-peak gain from $\w$ to $\eznf$ is less than $\gamma$.
A controller $\K$ satisfying Assumption~\ref{ass:p2p_LMIs} and minimizing $\gamma$ can be designed using the algorithm in~\cite{Schwenkel2025}.
Given a bound on the initial estimation error, we can bound $\|\eznf_{k|t}\|_2^2$.
\begin{assumption}\label{ass:cf0}
	Let $\cf_{0|-1}\geq 0$ satisfy $\chi_{0|-1} ^\top \P \chi_{0|-1} \leq \cf_{0|-1}$.
\end{assumption}
%%%%%%%%%%%%%%%%%%%%%%%%%%%%%%%%%%%%%%%%%%%%%%%%%%%%%%%%%%%%%%%%%%%%%%%%%%%%%%%%%%%% THEOREM 1 %%%%%%%%%%%%%%%%%%%%%%%%%%%%%%%%%%%%%%%%%%%%%%%%%%%%%%%%%%%%%%%%%%%%%%%%%%
\begin{theorem}\label{thm:constraint_tightening_no_opt}
	Let Assumptions~\ref{ass:iqc},~\ref{ass:p2p_LMIs}, and~\ref{ass:cf0} hold and let $\GKxnf_{t|t} = \GKxnf_{t|t-1}$.
	Then, for all $t\in\N$, $k\geq t$ we have
	\begin{align}\label{eq:constraint_tightening}
		\|\eznf_{k|t}\|_2^2 \leq {\textstyle \frac{\gamma}{\alpha}}\cf_{k|t} + {\textstyle \frac{\gamma}{\alpha}}\beta(\|\qnf_{k|t}\|_2^2+\dmax^2)
	\end{align}
	with $\cf_{t|t} = \cf_{t|t-1}$ and
	\begin{align}\label{eq:cf}
		\cf_{k+1|t} & = \rho^2 \cf_{k|t} + \mu \rho^2\big(\|\qnf_{k|t}\|_2^2 + \dmax^2 \big) .
	\end{align}
\end{theorem}
The proof can be found in the Appendix~\ref{sec:proofs}.
Note that~\eqref{eq:constraint_tightening}--\eqref{eq:cf} imply $\limsup_{k\to\infty} \|\eznf_{k|t}\|_2^2\leq \gamma^2(\dmax^2+\|\qnf\|_\peak^2)$.
Hence, by designing $\K$ such that it minimizes $\gamma$, we minimize the tube size.
Due to the initialization $\GKxnf_{t|t} = \GKxnf_{t|t-1}$ in Theorem~\ref{thm:constraint_tightening_no_opt}, there is no feedback from the measurements $\y_t$ to the nominal trajectory $\GKxnf_{t|t}$, similar to~\cite{Mayne2001,Schwenkel2020}.
Next, we show how an initial condition $\GKxnf_{t|t}$ and a smaller error bound $\cf_{t|t}$ can be computed based on an estimate $\GKxe_t$ and a bound $\ce_t$ on the estimation error $\eGKxe_t$ and the IQC.
\begin{assumption}\label{ass:obs_bound}
	Let $\ubar \chi_t\!\triangleq\!\begin{bmatrix} \filtx_t; \rho^{-2t}\eGKxe_t \end{bmatrix}$. For all $t\in\N$ there exists a known bound $\ce_t\geq 0$ satisfying 
	\begin{align}\label{eq:eGKxe_bound}
		\rho^{-2t}\ce_{t}&\geq \ubar \chi_t^\top \P \ubar\chi_t + \sum_{j=0}^{t-1} \filty_j^\top \M \filty_j.
	\end{align}
\end{assumption}
In Section~\ref{sec:estimator}, we show how to design a robust estimator  and a bound $\ce_t$ that satisfy Assumption~\ref{ass:obs_bound}.
Using this additional information, we initialize 
\begin{subequations}\label{eq:init_opt}
	\begin{align}
		\GKxnf_{t|t}&=\initopt_t \GKxnf_{t|t-1} + (1-\initopt_t) \GKxe_t\\
		\cf_{t|t}  & = \initopt_t \cf_{t|t-1} + (1-\initopt_t) \ce_t
	\end{align}	
\end{subequations}
where  $\initopt_t\in[0,1]$ is a decision variable of the MPC scheme that interpolates between the prediction and the estimate, similar to~\cite{Schlueter2022,Koehler2022}. 
\begin{theorem}\label{thm:constraint_tightening}
	Let Assumptions~\ref{ass:iqc},~\ref{ass:p2p_LMIs},~\ref{ass:cf0}, and~\ref{ass:obs_bound} hold. 
	Then, for all $t\in \N$, $k\geq t$, and $\initopt_t\in[0,1]$ we have \eqref{eq:constraint_tightening} with \eqref{eq:cf},~\eqref{eq:init_opt}.
\end{theorem}
The proof can be found in the Appendix~\ref{sec:proofs}.

\section{Robust model predictive control scheme}\label{sec:rmpc}
In this section, we define the MPC optimization problem, which exploits the bound~\eqref{eq:constraint_tightening} on the prediction error to tighten the constraints accordingly.
At each time step $t$, given $\GKxe_t$, $\ce_t$, $\GKxnf_{t|t-1}$, and $\cf_{t|t-1}$, we compute $\initopt_t$ and $\unf_{\cdot | t}$ by solving the following optimization problem
\begin{subequations}\label{eq:mpc}
	\begin{align}
		&\min_{\unf_{\cdot|t},\initopt_t} \ \!\! \sum_{k=t}^{t+N-1} \left\| \begin{pmatrix} \GKxnf_{k|t}  \\ \unf_{k|t} \end{pmatrix} \right\|_\scw^2\!\! + \|\GKxnf_{t+N|t}\|_\tcw^2 \triangleq J(\GKxnf_{t|t},\unf_{\cdot|t}) \\
		&\text{s.t.}\  \text{nominal and tube dynamics~\eqref{eq:xnf_sys},~\eqref{eq:cf},~\eqref{eq:init_opt}} \\
		&\ \ \znf_{k|t} \leq 1 -\sqrt{\frac{\gamma}{\alpha}\cf_{k|t} + \frac{\gamma}{\alpha}\beta (\|\qnf_{k|t}\|_2^2+\dmax^2)} \ \forall k \in \I_{[t,t+N-1]}\label{eq:mpc:constraint} \\
		&\ \ \|\tswroot\GKxnf_{t+N|t}\|_{2}^2 \leq \GKxnf^\mathrm{f},\quad \cf_{t+N|t} \leq \cf^\mathrm{f}\label{eq:mpc:term}
	\end{align}	
\end{subequations}
where the cost weighting matrix $\scw\succ 0$ is a design parameter that can be tuned to achieve secondary performance goals beyond stability and constraint satisfaction.
The constraint~\eqref{eq:mpc:constraint} corresponds to $\znf_{k|t}\leq 1-\eznf_{k|t}$ with the upper bound~\eqref{eq:constraint_tightening}.
Stability and constraint satisfaction are ensured by a suitable choice of the terminal cost matrix $\tcw$ and the terminal set~\eqref{eq:mpc:term} defined by $\tswroot$, $\GKxnf^\mathrm{f}$, $\cf^\mathrm{f}$. 
\begin{assumption}\label{ass:terminal}
	Let $\tswroot^\top \tswroot \succ 0 $, $\tcw\succ 0$, $\GKxnf^\mathrm{f}\geq 0$, and $\cf^\mathrm{f}\geq 0$ satisfy
	\begin{subequations}
	\begin{align}
		&\A\GK^\top \tcw \A\GK - \tcw \preceq -\begin{pmatrix}I & 0\end{pmatrix} \scw \begin{pmatrix}I&0\end{pmatrix}^\top  \label{eq:ass_tcw} \\
		&\A\GK^\top \tswroot^\top \tswroot \A\GK - \tswroot^\top \tswroot \preceq 0 \label{eq:ts_inv1}\\
		&\cf^\mathrm{f} (1-\rho^2)\geq  \mu\rho^2 (\gamma_\w^\mathrm{f})^2 \label{eq:ts_inv2} \\
		&\|\C\GK\z\tswroot^{-1} \|_{2\shortrightarrow\infty} \sqrt{\GKxnf^\mathrm{f}} \leq 1 - \sqrt{\frac{\gamma}{\alpha} \cf^\mathrm{f}+\frac{\gamma}{\alpha}\beta(\gamma_\w^\mathrm{f})^2} \label{eq:ts_feas}
	\end{align}
	where $(\gamma_\w^\mathrm{f})^2\triangleq \|\C\GK\q\tswroot^{-1}\|_2^2\GKxnf^\mathrm{f} + \dmax^2$ and the induced $2$-to-$\infty$ matrix norm is defined by $\|A\|_{2\shortrightarrow\infty} \triangleq \max_{\|x\|_2= 1 } \|Ax\|_\infty$.\footnote{Can be computed via $\|A\|_{2\shortrightarrow\infty} = \max_{i} \|e_i^\top A \|_2$, where $e_i\in\R^\nz$ is the $i$-th unit vector, i.e., $I_\nz = \begin{pmatrix}e_1&\dots&e_\nz\end{pmatrix}$.}
	\end{subequations}
\end{assumption}
Inequalities~\eqref{eq:ts_inv1},~\eqref{eq:ts_inv2} correspond to invariance of the terminal set, while~\eqref{eq:ts_feas} is needed for feasibility of the terminal set. 
The inequality~\eqref{eq:ass_tcw} is stating that the terminal cost is an upper bound on the cost-to-go.
The following Lemma provides a sufficient condition to satisfy these inequalities.
\begin{lemma}\label{lem:ts}
	Suppose that $\gamma\dmax \leq 1$. Then there exist $\tcw$, $\tswroot$, $\cf^\mathrm{f}$, and $\GKxnf^\mathrm{f}$ such that Assumption~\ref{ass:terminal} holds.
\end{lemma}
\begin{proof}
	As $\A\GK$ is Schur stable (follows from~\eqref{eq:stab_LMI1},~\eqref{eq:stab_LMI2} as shown in~\cite{Schwenkel2025}), we can choose $\tcw\succ 0$ satisfying~\eqref{eq:ass_tcw} as the unique solution of the Lyapunov equation $\A\GK^\top \tcw \A\GK - \tcw = -\begin{pmatrix}I & 0\end{pmatrix} \scw \begin{pmatrix}I&0\end{pmatrix}^\top\prec0$. 
	Condition~\eqref{eq:ts_inv1} follows from~\eqref{eq:ass_tcw} when choosing $\tswroot$ as $\tswroot^\top \tswroot=\tcw$.
	Further, set $\GKxnf^\mathrm{f} = 0$ and $\cf^\mathrm{f} = \frac{\mu \rho^{2}}{1-\rho^2} \dmax^2= \alpha \mu \dmax^2$.
	Then,~\eqref{eq:ts_inv2} holds and due to $\gamma\dmax \leq 1$ we have $1 - \sqrt{\gamma^2\dmax^2} \geq 0$ which yields~\eqref{eq:ts_feas}.
\end{proof}
The condition $\gamma\dmax\leq 1$ ensures that the origin lies in the tightened constraint set, which imposes a maximal bound on the disturbance bound $\dmax$.
Next, we provide closed-loop guarantees for the proposed MPC scheme.
%%%%%%%%%%%%%%%%%%%%%%%%%%%%%%%%%%%%%%%%%%%%%%%%%%%%%%%%%%%%%%%%%%%%%%%%%%%%%%%%%%%% THEOREM 2 %%%%%%%%%%%%%%%%%%%%%%%%%%%%%%%%%%%%%%%%%%%%%%%%%%%%%%%%%%%%%%%%%%%%%%%%%%
\begin{theorem}\label{thm:iss}
	Let Assumptions~\ref{ass:iqc},~\ref{ass:p2p_LMIs},~\ref{ass:cf0},~\ref{ass:obs_bound} and~\ref{ass:terminal} hold.
	Then, the MPC problem~\eqref{eq:mpc} is recursively feasible, i.e., if~\eqref{eq:mpc} is feasible at time $t=0$ exists, then it is feasible for all $\t\in\N$.
	Further, the closed loop robustly satisfies the constraints~\eqref{eq:constraints} and is input-to-state stable, i.e., there exist $\beta_1\in\KL$ and $\alpha_1\in\Kinf$ such that for all $t\in\N$ we have
	\begin{align}\label{eq:iss}
		\| \GKx_t \|^2_2 &\leq \beta_1\left(\left\|\x_0 \right\|_2+\left\|\exnf_0\right\|_2,t\right) + \alpha_1(\|\d\|_\peak).
	\end{align}
\end{theorem}
The proof can be found in Section~\ref{sec:proofs}. 
As all statements of the theorem are proven based on a feasible candidate  solution with $\initopt_t=1$, one can analogously show the guarantees from Theorem~\ref{thm:iss} for a fixed $\initopt_t = 1$ without requiring Assumption~\ref{ass:obs_bound}.
The input-to-state stability bound~\eqref{eq:iss} depends on the initial condition and the error in the initial estimate.

\section{Robust estimation with error bound}\label{sec:estimator}
In this section, we design an estimator providing an estimate $\GKxe_t$ for $\GKx_t$ and an error bound $\ce_t$ satisfying Assumption~\ref{ass:obs_bound}. 
\arxivonly{Alongside we provide a convex solution to the robust estimation problem via IQCs with peak-to-peak performance. }
As we compute the input $\u_t$ based on the estimate $\GKxe_t$, which is computed based on the measurement $\y_t$, we require that $\y_t$ is independent of $\u_t$. 
\begin{assumption}[No feedthrough from $\u$ to $\y$]\label{ass:no_alg_loop}
	Let $a,u\in\ell_{2,\mathrm{e}}^\nq$, $t\in\N$, define $u^{(t-1)}$ to be the signal $\u$ truncated at time $t-1$, i.e., $u^{(t-1)}_k = u_k$ for all $k\in \I_{[0,t-1]}$ and $u^{(t-1)}_k = 0$ for $k\geq t$.
	Then, for all such $a,u,t$ we have $\D\G\y\p \big(\Delta(a+\D\G\q\u \u)\big)_{t} = \D\G\y\p \big( \Delta(a+\D\G\q\u \u^{(t-1)} )\big)_{t}$.
\end{assumption}
The proposed estimator estimates the error to a (known) nominal system
\begin{flalign}\label{eq:GKxnp_sys}
	\left[\GKxnp_{t+1};\, \qnp_{t};\, \bullet;\, \ynp_t \right] &=
	\Theta \left[\GKxnp_{t};\, 0;\, 0;\, \unf_{t}\right], \ \GKxnp_0=[\xnf_0;\,0].\hspace{-1cm}&
\end{flalign}
The error between the nominal and the actual system is defined by $\eGKxnp \triangleq \GKx-\GKxnp$, $\eynp\triangleq \y-\ynp$, and $\eqnp\triangleq \q-\qnp$.
Due to~\eqref{eq:DGKx_sys},~\eqref{eq:GKxnp_sys}, the error dynamics are
	\begin{align}\label{eq:eGKxnp_sys}
		\left[\eGKxnp_{t+1};\, \eqnp_{t};\, \bullet;\, \eynp_t \right] &=
		\Theta \left[\eGKxnp_{t};\, \p_t;\, \d_t;\, 0 \right]
	\end{align}
with the initial condition $\eGKxnp_0=\GKx_0-\GKxnp_0$ and where $\p=\Delta(\q)=\Delta(\qnp+\eqnp)$.
The estimator $\L$ with state $\Lx$ uses~$\eynp_t$ and~$\qnp_t$ to compute an estimate $\eGKxnpe_t$ for the state $\eGKxnp_t$
\begin{subequations}\label{eq:obs_dyn}
	\begin{align}
		\Lx_{t+1}    & = \A\L\Lx_t + \BB\L \eynp_t \\ 
		\eGKxnpe_{t} & = \CC\L\Lx_t + \DD\L \eynp_t  
	\end{align}
\end{subequations}
with initial condition $\Lx_0=0\in\R^{\nL}$.
The estimate $\GKxe_t$ for $\GKx_t$ can then be computed by $\GKxe_t \triangleq \GKxnp_t + \eGKxnpe_t$.
Hence, the estimation error satisfies $\eGKxe_t = \GKx_t-\GKxe_t =  \eGKxnp_t-\eGKxnpe_t$.
To establish~\eqref{eq:eGKxe_bound} we choose the performance channel $\wobs_t \triangleq \begin{bmatrix} \qnp_t;\, \d_t \end{bmatrix}$ to $\zobs_t \triangleq \rho^t \ubar\chi_t = [\rho^t\filtx_t;\,\eGKxnp_t]-[0;\,\eGKxnpe_t]$.
The augmented state $\xi_t=\begin{bmatrix} \filtx_t;\, \overline{\eGKxnp}{}_t;\, \lt \Lx_t \end{bmatrix}$, where $\overline{\eGKxnp}_t \triangleq \rho^{-t}\eGKxnp_t $, $\lt \Lx_t \triangleq \rho^{-t} \Lx_t$ follows 
\begin{align}\label{eq:xi_sys}
	\left[\xi_{t+1};\,\filty_t;\,\barzobs_t\right] = \obsrho \left[\xi_t;\,\lt\p_t;\,\barwobs_t\right]
\end{align}
with initial condition $\xi_0 = [0;\,\eGKxnp_0;\,0]$ and where the matrix $\obsrho$ is defined in~\eqref{eq:ABCDobs} in the appendix.
The following theorem provides a performance bound of an estimator $\L$.

%%%%%%%%%%%%%%%%%%%%%%%%%%%%%%%%%%%%%%%%%%%%%%%%%%%%%%%%%%%%%%%%%%%%%%%%%%%%%%%%%%%% THEOREM 3 %%%%%%%%%%%%%%%%%%%%%%%%%%%%%%%%%%%%%%%%%%%%%%%%%%%%%%%%%%%%%%%%%%%%%%%%%%
\begin{theorem}\label{thm:obs_p2p}
	Let Assumptions~\ref{ass:iqc} and~\ref{ass:no_alg_loop} hold.
	Suppose there exist $(\M_3,\X_3)\in \MXset$, $(\M_4,\X_4)\in\MXset$, $\Pobs={\Pobs}^\top$, and $\gamobs\geq \muobs\geq 0$ satisfying
	\begin{align}\label{eq:obs_LMI1}
		&\!\!\!\symb {
			{\arraycolsep=2pt\diagmat{-\Pobs                                      \\&\Pobs\\&&\Mobs\\&&& -\muobs I}\hspace{-0.1cm}
		\begin{pmatrix}
				I            & 0                & 0                \\
				\A\obsrho    & \B\obsrho\p      & \B\obsrho\w      \\[0.5mm]
				\C\obs\filty & \D\obs{\filty}\p & \D\obs{\filty}\w \\
				0            & 0                & I
			\end{pmatrix}}}\prec 0\\\nonumber
		&\!\!\!\symb{
			{\arraycolsep=1pt
					\begin{pmatrix}
						\ubar\X_3-\Pobs                                     \\
						 & \ubar\X_4                                        \\
						 &           & \M_4                                 \\
						 &           &      & P                             \\
						 &           &      &   & -\beta^\mathrm{o} I
					\end{pmatrix}
				\hspace{-0.1cm}}
				{\arraycolsep=2pt
			\begin{pmatrix}
				I            & 0              & 0              \\[0.3mm]
				\A\obsrho    & \B\obsrho\p    & \B\obsrho\w    \\[0.8mm]
				\C\obs\filty & \D\obs\filty\p & \D\obs\filty\w \\[0.8mm]
				\C\obs\z     & \D\obs\z\p     & \D\obs\z\w     \\[0.5mm]
				0            & 0              & I
			\end{pmatrix}}
		} \prec 0\\[-0.5cm]
		\label{eq:obs_LMI2}
	\end{align}
	with $\Mobs=\M+\M_3+\M_4$, $\alpha = \frac{\rho^2}{1-\rho^2}$, $\beta^\mathrm{o}=\alpha(\gamobs-\muobs)$, and $\ubar\X_i = \diag(\X_i,0)$ for $i\in\{3,4\}$.
	Further, let $\cp_0\in\R$ satisfy $\cp_0 \geq \xi_0^\top  \Pobs \xi_0$.
	Then, Assumption~\ref{ass:obs_bound} holds with
	\begin{subequations}\label{eq:cep}
		\begin{align}\label{eq:ce}
			\ce_t     & = \cp_{t} + \beta^\mathrm{o} (\|\mathscr{C} \GKxnp_t\|^2_2+\dmax^2)       \\
			\cp_{t+1} & = \rho^2\cp_{t} + \muobs  \rho^2\big(\|\qnp_t\|_2^2 + \dmax^2\big) \label{eq:cp}
		\end{align}
	\end{subequations}
	where $\mathscr{C}\triangleq (I-\mathscr{D}(\mathscr{D}^\top \mathscr{D})^{-1}\mathscr{D}^\top)\C\GK\q$ and where the columns of $\mathscr{D}$ are a basis of the image of $\D\G\q\u$.
	Moreover, if $\cp_0\leq \alpha \muobs (\|\qnp\|_\peak^2+\dmax^2)$, then for all $t\in \N$ we have
	\begin{align}\label{eq:p2p_obs}
		\ce_t & \leq \alpha \gamobs (\|\qnp\|_\peak^2+\dmax^2).
	\end{align}
\end{theorem}
The proof can be found in the Appendix~\ref{sec:proofs}.
The bound~\eqref{eq:p2p_obs} tells us that it is desirable to minimize $\gamobs$ in order to obtain small $\ce_t$.
\cdconly{In an extended version~\cite{Schwenkel2025b} of this paper, we show that robust estimation problem using IQCs with peak-to-peak performance can be convexified, which can be used to compute $\L$ satisfying~\eqref{eq:obs_LMI1}--\eqref{eq:obs_LMI2} and minimizing $\gamobs$.
}
\arxivonly{
Next, we describe for given $\K$, $\M$, and $\P$ how to find $\L$, $\M_3$, $\M_4$, $\X_3$, $\X_4$, and $\muobs$ that minimize $\gamobs$.
While the controller design in~\cite{Schwenkel2025} requires an iterative algorithm, we show that the problem of robust estimation with minimal peak-to-peak gain bound $\gamobs$ can be reformulated as a single convex SDP, similar to the continuous time robust estimation for $\mathcal H_\infty$ performance~\cite{Scherer2008}.

%%%%%%%%%%%%%%%%%%%%%%%%%%%%%%%%%%%%%%%%%%%%%%%%%%%%%%%%%%%%%%%%%%%%%%%%%%%%%%%%%%%% THEOREM 4 %%%%%%%%%%%%%%%%%%%%%%%%%%%%%%%%%%%%%%%%%%%%%%%%%%%%%%%%%%%%%%%%%%%%%%%%%%
\begin{theorem}\label{thm:obs_syn}
	Let~\eqref{eq:stab_LMI2} hold, let $\nL = \nx+\nK+\nfilt$, and decompose $\P=\begin{pmatrix} \P_{11} & \P_{12} \\ \P_{21} & \P_{22} \end{pmatrix}$, $\Pobs = \begin{pmatrix} \Pobs_{11} & \Pobs_{12} \\ \Pobs_{21} & \Pobs_{22} \end{pmatrix}$ with $\P_{11} \in\S^\nfilt$, $\P_{22}\in\S^{\nx+\nK}$, $\Pobs_{11},\Pobs_{22}\in\S^{\nL}$.
	Then, there exist $\Pobs$, $\gamobs$, $\muobs$, $\M_3$, $\M_4$, $\X_3$, $\X_4$, $\A\L$, $\BB\L$, $\CC\L$, $\DD\L$ satisfying the matrix inequalities~\eqref{eq:obs_LMI1},~\eqref{eq:obs_LMI2}, and $\Pobs_{22} \succ 0$ if and only if there exist $\Psynobs_{1}$, $\Psynobs_{2}$, $\gamobs$, $\muobs$, $\M_3$, $\M_4$, $\X_3$, $\X_4$, $\Ksynobs$, $\Lsynobs$, $\Msynobs$, and $\Nsynobs$ satisfying
	\begin{align} \label{eq:obs_LMI1_syn}
		 & \begin{pmatrix}
			   \mathcal{O}_{1} & \star \\  \begin{pmatrix}\Asynobs_1     & \Asynobs_2  & \Bsynobs_1       & \Bsynobs_2\end{pmatrix} & \Psynobs_{1}- \Psynobs_{2}
		   \end{pmatrix} \prec 0 \\ \label{eq:obs_LMI2_syn}
		 & \begin{pmatrix}
			   \mathcal{O}_{2} & \star \\ \begin{pmatrix}
				\Csynobs_1 & \Csynobs_2 & \Dsynobs_{1} & \Dsynobs_2 \\
			\end{pmatrix} & -\P_{22}^{-1}
		   \end{pmatrix} \prec 0
	\end{align}
	with $\nw = \nd+\nq$ and
	\begin{align*}
		\mathcal{O}_1\! & \triangleq \diag\left(
			\begin{pmatrix}
				-\Psynobs_{1} & -\Psynobs_{1}                                        \\
				-\Psynobs_{1} & -\Psynobs_{2}                                        \\
				\end{pmatrix},
			\begin{pmatrix}
					0 & 0 \\
					0 & 0
				\end{pmatrix}\right)                                                   \\
			& \quad + \symb{\begin{pmatrix}
				\Psynobs_{2} & I                     \\
				I            & 0                     \\
			\end{pmatrix} \begin{pmatrix} 
			\A\gpobsrho    & \A\gpobsrho & \B\gpobsrho\p    & \B\gpobsrho\w    \\
			\Asynobs_1     & \Asynobs_2  & \Bsynobs_1       & \Bsynobs_2       \\
			\end{pmatrix}} \\
				& \quad + \symbscalar{\Mobs \begin{pmatrix}
					\C\gpobs\filty & \C\gpobs\filty            & \D\gpobs\filty\p & \D\gpobs\filty\w 					
				\end{pmatrix}} - \muobs \diag(0,I_{\nw}) \\
		\mathcal{O}_2\! & \triangleq \diag\left(
		\begin{pmatrix}
			\ubar\X_3-\Psynobs_{1} & \ubar\X_3-\Psynobs_{1} \\
			\ubar\X_3-\Psynobs_{1} & -\Psynobs_{2}
			\end{pmatrix},
		\begin{pmatrix}
				0 & 0 \\
				0 & 0
			\end{pmatrix}\right)                                                   \\
		              & \quad + \symbscalar{\ubar\X_4
			\begin{pmatrix}
				\A\gpobsrho & \A\gpobsrho & \B\gpobsrho\p & \B\gpobsrho\w
			\end{pmatrix}}\\
			& \quad + \symb{
  \begin{pmatrix}
	  \P_{11} & \P_{12} \\
	  \P_{21} & 0
  \end{pmatrix}
\begin{pmatrix}
	  \begin{pmatrix}
		  I & 0
	  \end{pmatrix} &
	  \begin{pmatrix}
		  I & 0
	  \end{pmatrix} & 0          & 0                           \\
	  \Csynobs_1      & \Csynobs_2 & \Dsynobs_{1} & \Dsynobs_2\end{pmatrix}}\\
		& \quad + \symbscalar{\M_4
			\begin{pmatrix}
				\C\gpobs\filty & \C\gpobs\filty & \D\gpobs\filty\p & \D\gpobs\filty\w
			\end{pmatrix}
		} -\beta^\mathrm{o} \diag(0,I_{\nw})                                                                     
	\end{align*}
	and $\Asynobs_1=\Lsynobs\C\gpobs\y + \Ksynobs$, $\Asynobs_2=\Lsynobs\C\gpobs\y$, $\Bsynobs_1=\Lsynobs\D\gpobs\y\p$, $\Bsynobs_2=\Lsynobs\D\gpobs\y\w$, $\Csynobs_1=\begin{pmatrix}
		0 & I
	\end{pmatrix}
-\Nsynobs\C\gpobs\y - \Msynobs$, $ \Csynobs_2=
\begin{pmatrix}
	0 & I
\end{pmatrix}
-\Nsynobs\C\gpobs\y$, $\Dsynobs_{1}=-\Nsynobs\D\gpobs\y\p$, and $\Dsynobs_2=-\Nsynobs\D\gpobs\y\w$.
	In particular, a solution of~\eqref{eq:obs_LMI1},~\eqref{eq:obs_LMI2} is given by
	\begin{flalign}\label{eq:Pobs_syn}
		\Pobs            & =
		\begin{pmatrix}
			\Psynobs_{2} & I \\ I & (\Psynobs_{2} -\Psynobs_{1})^{-1}
		\end{pmatrix} \\\label{eq:L_syn}
		\begin{pmatrix}
			\A\Lrho  & \BB\Lrho \\
			\CC\L & \DD\L
		\end{pmatrix} & =
		\begin{pmatrix}
			\Ksynobs & \Lsynobs \\
			\Msynobs & \Nsynobs
		\end{pmatrix}
		\begin{pmatrix}
			(\Psynobs_{1} - \Psynobs_{2})^{-1} & 0 \\
			0                                  & I
		\end{pmatrix}.\hspace{-2cm}&
	\end{flalign}
\end{theorem}
The proof can be found in the Appendix~\ref{sec:proofs}.

While Theorem~\ref{thm:obs_p2p} and~\ref{thm:obs_syn} is rather targeted to the specific performance output $\zobs_t=\rho^t \ubar\chi_t$ and the specific bound~\eqref{eq:eGKxe_bound}, we can want to briefly discuss we can analogously solve standard robust estimation problem using IQCs where we minimize the peak bound on the estimation error.
\begin{corollary}\label{cor:p2p_obs}
	Consider the assumptions of Theorem~\ref{thm:obs_p2p} with $\P_{11}=0$, $\P_{12}=\P_{21}^\top=0$, $\M=0$ and $\P_{22}=\alpha I$ and let $\dmax = \|\d\|_\peak$.
	Then $\|\eGKxe\|_\peak^2 \leq \gamobs \|\wobs\|_\peak^2$.
\end{corollary}
\begin{proof}
	The statement follows directly by using the conditions~\eqref{eq:eGKxe_bound} and~\eqref{eq:p2p_obs}, which yields for all $t\geq 0$ that ${\zobs_t}^\top \P\zobs_t = \alpha\|\eGKxe_t\|^2_2 \leq \alpha \gamobs \|\wobs\|_\peak^2$.
\end{proof}
Hence, we can also apply Theorem~\ref{thm:obs_syn} to design an estimator while minimizing the bound $\gamobs$ on the peak-to-peak gain from $\wobs$ to the estimation error $\eGKxe$. 
Other performance inputs or outputs can be included analogously.
}

\section{Implementation}
Summarizing the previous sections, to implement the proposed MPC scheme, we need to perform the following offline design steps: 
1. design $\K$ such that Assumption~\ref{ass:p2p_LMIs} holds with minimal $\gamma$ by following the algorithm from~\cite{Schwenkel2025};
\arxivonly{2. design $\L$ by minimizing~$\gamobs$ subject to \eqref{eq:obs_LMI1_syn} and~\eqref{eq:obs_LMI2_syn}; }
\cdconly{2. design $\L$ such that \eqref{eq:obs_LMI1} and~\eqref{eq:obs_LMI2} hold; }
3. design terminal conditions that satisfy Assumption~\ref{ass:terminal}.
After the offline design is completed, we can apply the scheme and execute the following steps online at each time $t$: 
\begin{enumerate}
	\item measure $\y_t$
	\item compute the estimate $\GKxe_t$ and $\ce_t$ via \eqref{eq:obs_dyn} and \eqref{eq:cep}
	\item compute $\unf_{t|t}$ by solving the MPC problem~\eqref{eq:mpc}
	\item apply the input $\u_t$ from~\eqref{eq:mpc_input}.
\end{enumerate}
\begin{remark}\label{rem:ts}
	While Lemma~\ref{lem:ts} provides a constructive proof, it is possible to design a larger terminal region. 
	To this end, we fix some ratio $r^\mathrm{f}\geq 1 $ and set $\cf^\mathrm{f} = \frac{\mu \rho^{2}}{1-\rho^2} \big(\dmax^2+r^\mathrm{f}\|\C\GK\q\tswroot^{-1}\|_2^2\GKxnf^\mathrm{f}\big)$.
	Due to $r^\mathrm{f}\geq 1$, the inequality~\eqref{eq:ts_inv2} holds. 
	Then, we maximize $\GKxnf^\mathrm{f}$ subject to~\eqref{eq:ts_feas}. 
	Note that $\tswroot$ does not need to be chosen as $\tswroot^\top \tswroot=\tcw$.
\end{remark}
\begin{remark}[Improvements]\label{rem:cfi}
	After $\K$ and $\L$ are designed, one can perform a joint analysis of the corresponding $\gamma$ and $\gamobs$ by minimizing $\gamma+\gamobs$ (or some weighted sum) subject to~\eqref{eq:stab_LMI1},~\eqref{eq:ana_LMI2},~\eqref{eq:stab_LMI2},~\eqref{eq:obs_LMI1},~\eqref{eq:obs_LMI2}. 
	The decision variables of this minimization are $\P$, $\Pobs$, $\mu$, $\muobs$, $\gamma$, $\gamobs$, $(\M_i,\X_i)_{i\in\{1,2,3,4\}}$. 
	Due to the nonlinear dependency on $\rho$, this variable is minimized via a line search procedure.
	Thereby, much smaller values of $\gamobs$ can be achieved for the price of a slightly larger $\gamma$.
	An additional modification to achieve significantly smaller tubes is to analyze the constraints componentwise.
	In particular, this means that for each $i\in\I_{[1,\nz]}$, we minimize $\gamma_i+\gamobs_i$ (or some weighted sum) subject to~\eqref{eq:stab_LMI1},~\eqref{eq:ana_LMI2},~\eqref{eq:stab_LMI2},~\eqref{eq:obs_LMI1},~\eqref{eq:obs_LMI2} where we replaced $\C\G\z$ and $\D\G\z j$ for all $j\in\{\p,\w,\u\}$ by $e_i^\top \C\G\z$ and $e_i^\top \D\G\z j$, where $e_i\in\R^\nz$ is the $i$-th unit vector, i.e., $I_\nz = \begin{pmatrix}e_1&\dots&e_\nz\end{pmatrix}$.
	Then, we obtain individual $\cf^{(i)}_{k|t}$ and $\cp_t^{(i)}$ for each $i\in\I_{[1,\nz]}$ and the constraint~\eqref{eq:mpc:constraint} in the MPC scheme becomes
	\begin{align*}
		e_i^\top \znf_{k|t} \leq 1 -\sqrt{\frac{\gamma_i}{\alpha_i}\cf_{k|t}^{(i)} + \frac{\gamma_i}{\alpha_i}\beta_i (\|\qnf_{k|t}\|_2^2+\dmax^2)}
	\end{align*}
	where $\alpha_i=\frac{\rho^2_i}{1-\rho^2_i}$, $\beta_i=\alpha_i(\gamma_i-\mu_i)$.
	By this approach, the shape of the resulting tube is not ellipsoidal anymore but similar to the polytopic shape of the constraint set, which reduces the conservatism significantly.
	Similarly, we obtain $\nz$ terminal constraints $\cf_{t+N|t}^{(i)} \leq \cf^\mathrm{f}_i$, but we stick to one $\|\tswroot \GKxnf_{t+N|t}\|_2\leq \GKxnf^\mathrm{f}$.
	To find $\GKxnf^\mathrm{f}$ and  $\cf^\mathrm{f}_i$, we adapt the procedure from Remark~\ref{rem:ts} such that we maximize $\GKxnf^\mathrm{f}$ subject to~\eqref{eq:ts_feas} with $\gamma_i$, $\alpha_i$, and $\beta_i$ for all $i\in\I_{[1,\nz]}$.
	Then, we set $\cf^\mathrm{f}_i = \frac{\alpha_i}{\gamma_i} (1-\|\C\GK\z\tswroot^{-1} \|_{2\shortrightarrow\infty} \sqrt{\GKxnf^\mathrm{f}})^2- \beta_i(\dmax^2+\|\tswroot \GKxnf_{t+N|t}\|_2\leq \GKxnf^\mathrm{f})$. 
\end{remark}
\begin{example}
	Consider the example from~\cite{Schwenkel2020} to highlight that the estimator-based output feedback approach has benefits even if the state $\x_t$ is fully known.
	The system is given by
	\begin{align*}
		\resizebox{\linewidth}{!}{\small$\blkmat{c:c:c:c}{\A\G&\B\G\p&\B\G\d&\B\G\u\newblkdash \C\G\q&\D\G\q\p&\D\G\q\d&\D\G\q\u\newblkdash \C\G\z&\D\G\z\p&\D\G\z\d&\D\G\z\u\newblkdash \C\G\y&\D\G\y\p&\D\G\y\d&\D\G\y\u }
		\!\!=\!\!\scriptsize\blkmat{cc:c:c:c}{\!.995 & .095 & .005 & .002 & .005 \\\! -.095 & .900 & .095 & .038& .095 \newblkdash 1 & 1 & 0 & 0 & -1 \newblkdash 1 & 0 & 0 & 0 & 0 \\ 0 & 1 & 0 & 0 & 0 \\ 0 & 0 & 0 & 0 & 1\newblkdash 1 & 0 & 0 & 0 & 0 \\ 0 & 1 & 0 & 0 & 0 }$}
	\end{align*}
	with the constraints $\z_{1,t}=\x_{1,t}\in[-.1,1]$, $\z_{2,t}=\x_{2,t}\in[-.25,.05]$, $\z_{3,t}=\u\in[-1,1]$, where $\z_{i,t}$ denotes the $i$-th component of the vector $\z_t$.
	The (loop-transformed) uncertainty $\Delta_\rho$ satisfies the $\mathcal H_\infty$ norm bound $\|\Delta_\rho\|_\infty\leq 0.2285$ for all $\rho \in [0.85,1]$.
	Hence, Assumption~\ref{ass:iqc} is satisfied for all $\rho \in [0.85,1)$ for the finite horizon IQC with terminal cost for dynamic uncertainties from~\cite{Scherer2022a,Schwenkel2025}.
	The peak-to-peak minimization algorithm from~\cite{Schwenkel2025} yields a controller $\K$ with peak-to-peak gain $\gamma = 0.949$ for $\rho=0.85$. 
	The estimator synthesis (\arxivonly{Theorem~\ref{thm:obs_syn}}\cdconly{\cite[Theorem 5]{Schwenkel2025b}}) yields $\L$ with $\gamobs=0.949$.
	Following Remark~\ref{rem:cfi}, we obtain $\gamma_1 = 0.06$, $\gamobs_1=0.029$, $\gamma_2 = 0.727$, $\gamobs_2 = 0.452$, $\gamma_3=0.757$, and $\gamobs_3=0.472$, which is a significant reduction of the conservatism in the tube.
	All SDPs are solved using YALMIP~\cite{Lofberg2004} with Mosek~\cite{mosek}.
	The MPC problem is solved using CasADi~\cite{Andersson2019} with ipopt~\cite{Waechter2005}.
	The closed-loop trajectories are shown in Figure~\ref{fig:state_space} for the case with and without the estimator $\L$ (without $\L$, we fix $\initopt_t=1$ for all $t$).
	The estimator-based initial state optimization provides smaller tubes and faster convergence. 
	A quantitative comparison between these two approaches and the schemes from~\cite{Schwenkel2020} and~\cite{Schwenkel2023c} is given in Table~\ref{tab:exmp}.
	For the schemes from~\cite{Schwenkel2020} and~\cite{Schwenkel2023c} we use the state feedback controller $\K$ from~\cite{Schwenkel2020}, as they are not able to handle the dynamic controller $\K$ in \eqref{eq:prestab}. 
	The proposed scheme with initial state optimization via $\initopt_t$ in~\eqref{eq:init_opt} outperforms all other schemes, thereby showing that the estimator-based approach is beneficial even in the case where $\x_t$ is completely known.
	The code is available online\footnote{\href{https://github.com/Schwenkel/mpc-iqc}{https://github.com/Schwenkel/mpc-iqc}}.
	\begin{figure}
		\resizebox{\linewidth}{!}{\input{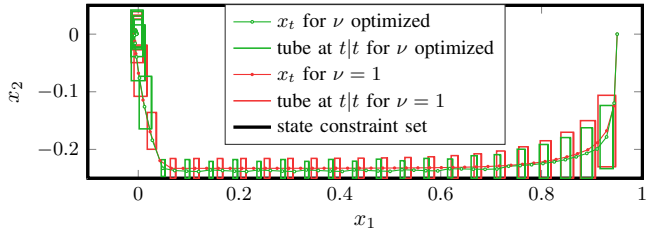}}\vspace{-0.3cm}
		\caption{Trajectories of the real system $\x_t$ together with the tube at $t|t$, which is for clarity only plotted at every second time instance. %The constraint set is the bounding box of the plot.
		}\label{fig:state_space}
	\end{figure}
	\begin{table}
		\centering
		\caption{Comparison of different IQC-based robust MPC schemes}\label{tab:exmp}
		\begin{tabular}{ccccc}
			\toprule
			MPC scheme & \cite{Schwenkel2020} & \cite{Schwenkel2023c}  & $\initopt=1$ & optimize $\initopt$ \\ \midrule
			closed-loop cost & $1676.4$ & $1510.3$ & $1441.1$ & $1414.6$ \\
			$\x_{1,t}$ at $t=43$ & $0.100$ & $0.081$ & $0.024$ & $0.012$ \\
			offline comp. time & 0.65s & 0.65s & 13.82s & 20.4s \\
			$\begin{matrix}
				\text{average online}\\\text{computation time}
			\end{matrix}$ & 112ms & 71ms & 95ms &129ms\\
			\bottomrule
		\end{tabular}
	\end{table}
\end{example}

\bibliographystyle{ieeetran}
\bibliography{../../../IST_Literatur/my_bib}

\begin{thebibliography}{10}
\providecommand{\url}[1]{#1}
\csname url@samestyle\endcsname
\providecommand{\newblock}{\relax}
\providecommand{\bibinfo}[2]{#2}
\providecommand{\BIBentrySTDinterwordspacing}{\spaceskip=0pt\relax}
\providecommand{\BIBentryALTinterwordstretchfactor}{4}
\providecommand{\BIBentryALTinterwordspacing}{\spaceskip=\fontdimen2\font plus
\BIBentryALTinterwordstretchfactor\fontdimen3\font minus \fontdimen4\font\relax}
\providecommand{\BIBforeignlanguage}[2]{{%
\expandafter\ifx\csname l@#1\endcsname\relax
\typeout{** WARNING: IEEEtran.bst: No hyphenation pattern has been}%
\typeout{** loaded for the language `#1'. Using the pattern for}%
\typeout{** the default language instead.}%
\else
\language=\csname l@#1\endcsname
\fi
#2}}
\providecommand{\BIBdecl}{\relax}
\BIBdecl

\bibitem{Rawlings2017}
J.~B. Rawlings, D.~Q. Mayne, and M.~M. Diehl, \emph{Model Predictive Control: Theory, Computation, and Design, 2nd Edition}.\hskip 1em plus 0.5em minus 0.4em\relax Nob Hill Publishing, LLC, 2017.

\bibitem{Megretski1997}
A.~Megretski and A.~Rantzer, ``System analysis via integral quadratic constraints,'' \emph{{IEEE} Trans. Automat. Control}, vol.~42, no.~6, pp. 819--830, 1997.

\bibitem{Veenman2016}
J.~Veenman, C.~W. Scherer, and H.~Köro{\v{g}}lu, ``Robust stability and performance analysis based on integral quadratic constraints,'' \emph{European J. Control}, vol.~31, pp. 1--32, 2016.

\bibitem{Scherer2022a}
C.~W. Scherer, ``Dissipativity and integral quadratic constraints: Tailored computational robustness tests for complex interconnections,'' \emph{{IEEE} Control Systems Magazine}, vol.~42, no.~3, pp. 115--139, 2022.

\bibitem{Chisci2001}
L.~Chisci, J.~A. Rossiter, and G.~Zappa, ``Systems with persistent disturbances: predictive control with restricted constraints,'' \emph{Automatica}, vol.~37, no.~7, pp. 1019--1028, 2001.

\bibitem{Mayne2001}
D.~Q. Mayne and W.~Langson, ``Robustifying model predictive control of constrained linear systems,'' \emph{Electronics Letters}, vol.~37, no.~23, pp. 1422--1423, 2001.

\bibitem{Schwenkel2025}
L.~Schwenkel, J.~Köhler, M.~A. Müller, C.~W. Scherer, and F.~Allgöwer, ``Multi-objective robust controller synthesis with integral quadratic constraints in discrete-time,'' \emph{arXiv:2503.22429}, 2025.

\bibitem{Scherer2008}
C.~W. Scherer and I.~E. Köse, ``Robustness with dynamic {IQCs}: An exact state-space characterization of nominal stability with applications to robust estimation,'' \emph{Automatica}, vol.~44, no.~7, pp. 1666--1675, 2008.

\bibitem{Schwenkel2020}
L.~{Schwenkel}, J.~{Köhler}, M.~A. {Müller}, and F.~{Allgöwer}, ``Dynamic uncertainties in model predictive control: Guaranteed stability for constrained linear systems,'' in \emph{Proc. 59th IEEE Conf. Decision and Control (CDC)}, 2020, pp. 1235--1241.

\bibitem{Schwenkel2023c}
L.~Schwenkel, J.~Köhler, M.~A. Müller, and F.~Allgöwer, ``Model predictive control for linear uncertain systems using integral quadratic constraints,'' \emph{{IEEE} Trans. Automat. Control}, vol.~68, pp. 355--368, 2023.

\bibitem{Lovaas2008}
C.~L{\o}vaas, M.~M. Seron, and G.~C. Goodwin, ``Robust output-feedback model predictive control for systems with unstructured uncertainty,'' \emph{Automatica}, vol.~44, no.~8, pp. 1933--1943, 2008.

\bibitem{Falugi2014}
P.~Falugi and D.~Q. Mayne, ``Getting robustness against unstructured uncertainty: A tube-based {MPC} approach,'' \emph{IEEE Trans. Automat. Control}, vol.~59, no.~5, pp. 1290--1295, 2014.

\bibitem{Loehning2014}
M.~Löhning, M.~Reble, J.~Hasenauer, S.~Yu, and F.~Allgöwer, ``Model predictive control using reduced order models: Guaranteed stability for constrained linear systems,'' \emph{J. Process Control}, vol.~24, no.~11, pp. 1647--1659, 2014.

\bibitem{Mayne2006}
D.~Q. Mayne, S.~V. Raković, R.~Findeisen, and F.~Allgöwer, ``Robust output feedback model predictive control of constrained linear systems,'' \emph{Automatica}, vol.~42, no.~7, pp. 1217--1222, 2006.

\bibitem{Subramanian2017}
S.~Subramanian, S.~Lucia, and S.~Engell, ``A novel tube-based output feedback {MPC} for constrained linear systems,'' in \emph{Proc. American Control Conference (ACC)}, 2017, pp. 3060--3065.

\bibitem{Schlueter2022}
H.~Schlüter and F.~Allgöwer, ``Stochastic model predictive control using initial state optimization,'' in \emph{Proc. 25th Int. Symp. Mathematical Theory of Networks and Systems (MTNS)}, vol.~55, no.~30.\hskip 1em plus 0.5em minus 0.4em\relax Elsevier {BV}, 2022, pp. 454--459.

\bibitem{Koehler2022}
J.~K{\"o}hler and M.~N. Zeilinger, ``Recursively feasible stochastic predictive control using an interpolating initial state constraint,'' \emph{{IEEE} Control Systems Letters}, vol.~6, pp. 2743--2748, 2022.

\bibitem{Scherer2018}
C.~W. Scherer and J.~Veenman, ``Stability analysis by dynamic dissipation inequalities: On merging frequency-domain techniques with time-domain conditions,'' \emph{Systems {\&} Control Letters}, vol. 121, pp. 7--15, 2018.

\bibitem{Hu2016}
B.~Hu and P.~Seiler, ``Exponential decay rate conditions for uncertain linear systems using integral quadratic constraints,'' \emph{{IEEE} Trans. Automat. Control}, vol.~61, no.~11, pp. 3631--3637, 2016.

\bibitem{Lofberg2004}
J.~Löfberg, ``{YALMIP} : {A} toolbox for modeling and optimization in {MATLAB},'' in \emph{Proc. {IEEE} Int. Conf. Robotics and Automation}, 2004.

\bibitem{mosek}
{MOSEK ApS}, \emph{The MOSEK optimization toolbox for MATLAB manual. Version 9.2.}, 2021, \url{http://docs.mosek.com/9.2/toolbox/index.html}.

\bibitem{Andersson2019}
J.~A.~E. Andersson, J.~Gillis, G.~Horn, J.~B. Rawlings, and M.~Diehl, ``{CasADi} -- {A} software framework for nonlinear optimization and optimal control,'' \emph{Mathematical Programming Computation}, vol.~11, no.~1, pp. 1--36, 2019.

\bibitem{Waechter2005}
A.~Wächter and L.~T. Biegler, ``On the implementation of an interior-point filter line-search algorithm for large-scale nonlinear programming,'' \emph{Mathematical Programming}, vol. 106, no.~1, pp. 25--57, 2005.

\end{thebibliography}

%%%%%%%%%%%%%%%%%%%%%%%%%%%%%%%%%%%%%%%%%%%%%%%%%%%%%%%%%%%%%%%%%%%%%%%%% PROOFS %%%%%%%%%%%%%%%%%%%%%%%%%%%%%%%%%%%%%%%%%%%%%%%%%%%%%%%%%%%%%%%%%%%%%%%%%%%%%%%
\renewcommand{\thesection}{A}
% redefine the command that creates the equation no.
%\setcounter{section}{0}  % reset counter 
\section{Appendix}\label{sec:proofs}
The interconnection of $\G$ and $\K$ in~\eqref{eq:GKx_sys} is given by
\begin{flalign} \nonumber
	&\GK \triangleq \begin{pmatrix}
		\A\GK & \B\GK p & \B\GK\d & \B\GK\u \\
		\C\GK\q & \D\GK\q\p & \D\GK\q\d & \D\GK\q\u \\
		\C\GK\z & \D\GK\z\p & \D\GK\z\d & \D\GK\z\u \\
		\C\GK\y & \D\GK\y\p & \D\GK\y\d & 0
	\end{pmatrix} \ \  \text{with} \ \  
	\blkmat{c:c}{\A\GK          & \B\GK i \newblkdash \C\GK j & \D\GK j i}
	\triangleq \\ \label{eq:GK}
	& \blkmat{cc:c}{
	\A\G + \B\G\u \DK \C\G\y    & \B{\G}\u\CK             & \B{\G} i + \B{\G}\u \DK \D\G\y i \\
	\BB{\K} \C\G\y              & \A{\K}                  & \BB{\K} \D\G\y i \newblkdash
	\C\G j +\D\G j\u \DK \C\G\y & \D\G j\u \CK            & \D\G j i+\D\G j \u \DK \D\G\y i} \hspace{-0.5cm}&
\end{flalign}
for all $i\in\{\p,\d,\u\}$, $j\in\{\q,\z,\y\}$. Combining $\GK$ with the filter $\filt$ as in~\eqref{eq:aug_sys} yields
\begin{align}\label{eq:ABCDgpobs}
	\resizebox{\linewidth}{!}{\arraycolsep=2pt$\!\blkmat{c:c:c}{
	\!\A\gpobsrho    & \B\gpobsrho\p           & \B\gpobsrho\w \!        \newblkdash
	\!\C\gpobs\filty & \D\gpobs\filty\p        & \D\gpobs\filty\w\!\newblkdash
	\!\C\gpobs\z     & \D\gpobs\z\p            & \D\gpobs\z\w\!\newblkdash
	\!\C\gpobs\y     & \D\gpobs\y\p            & \D\gpobs\y\w
	\!} \!\!\triangleq \!\! \blkmat{cc:c:cc}{
	\!\A\filt        & \B\filt\q\C\GK\q        & \B\filt\p+\B\filt\q\D\GK\q\p                & \B\filt\q       & \B\filt\q\D\GK\q\d \!               \\
	\!0              & \A\GKrho                & \B\GKrho\p                                  & 0               & \B\GKrho\d         \!  \newblkdash
	\!\C\filt\filty  & \D\filt\filty\q \C\GK\q & \D\filt\filty\p + \D\filt\filty\q \D\GK\q\p & \D\filt\filty\q & \D\filt\filty\q \D\GK\q\d\! \newblkdash
	\!0              & \C\GK\z                 & \D\GK\z\p                                   & 0               & \D\GK\z\d \! \newblkdash
	\!0              & \C\GK\y                 & \D\GK\y\p                                   & 0               & \D\GK\y\d \! 
	}\!$}
\end{align}
where $\A\GKrho = \rho^{-1}\A\GK$, $\B\GKrho i=\rho^{-1}\B\GK i$ for $i\in\{\p,\d,\u\}$.
Augmenting $\gpobsrho$ with the estimator $\L$ as in~\eqref{eq:xi_sys} yields
\begin{align} \label{eq:ABCDobs}
	\resizebox{\linewidth}{!}{$\underbrace{\!\!\!\blkmat{c:c:c}{
\A\obsrho                                       & \B\obsrho\p     & \B\obsrho\w       \newblkdash
\C\obs\filty                                    & \D\obs\filty\p  & \D\obs\filty\w  \newblkdash
\C\obs\z                                        & \D\obs\z\p      & \D\obs\z\w
}\!\!\!}_{ \triangleq\obsrho}\,
\triangleq\!\! \blkmat{cc:c:c}{
\A\gpobsrho                                     & 0               & \B\gpobsrho\p             & \B\gpobsrho\w                  \\
\BB\Lrho\C\gpobs\y                              & \A\Lrho         & \BB\Lrho \D\gpobs\y\p     & \BB\Lrho \D\gpobs\y\w  \newblkdash
\C\gpobs\filty                                  & 0               & \D\gpobs\filty\p          & \D\gpobs\filty\w  \newblkdash
\begin{pmatrix}I&\!\!\!0\end{pmatrix}                 & 0               & 0                         & 0                              \\
\begin{pmatrix}0&\!\!\!\!I\end{pmatrix}\!-\!\DD\L\C\gpobs\y & -\CC\L          &\! -\DD\L\D\gpobs\y\p         & \!-\DD\L\D\gpobs\y\p
}$}
\end{align}
with $\A\Lrho \triangleq \rho^{-1} \A\L$, $\BB\Lrho \triangleq \rho^{-1} \BB\L$.
\begin{proof}[Proof of Theorem~\ref{thm:constraint_tightening_no_opt}]
	We only prove Theorem~\ref{thm:constraint_tightening} as the special case $\initopt_t=1$ for all $t\in\N$ recovers Theorem~\ref{thm:constraint_tightening_no_opt}. 
\end{proof}
\begin{proof}[Proof of Theorem~\ref{thm:constraint_tightening}]
	We multiply~\eqref{eq:stab_LMI1} and~\eqref{eq:ana_LMI2} from the right and left by $\begin{bmatrix}\chi_{j|t};\, \lt\p_{j|t};\, \lt\w_{j|t}\end{bmatrix}$ and its transpose, yielding 
	\begin{align}\label{eq:delta1}
		\delta_1(j)&\leq 0 \quad \text{and} \quad \delta_2(j)\leq 0 \quad \text{with}\\
		\delta_1(j) & \triangleq(\star)^{\!\top}\!\P\chi_{j+1|t}-\chi_{j|t}^\top \P\chi_{j|t}+\filty_{j|t}^\top\M\filty_{j|t} - \mu \|\lt \w_{j|t}\|^2_2 \nonumber \\
		\delta_2(j) & \triangleq -\chi_{j|t}^\top \P\chi_{j|t}+\chi_{j|t}^\top \ubar \X_1\chi_{j|t}+(\star)^{\!\top}\! \ubar \X_2 \chi_{j+1|t}\nonumber\\
		&\ \quad +\filty_{j|t}^\top\M_2\filty_{j|t}+{\textstyle\frac{\alpha}{\gamma}}\|\overline{\eznf}_{j|t}\|^2_2  -\beta \|\lt \w_{j|t}\|^2_2.\nonumber
	\end{align}
	Due to~\eqref{eq:cf} and $\|\lt \w_{j|t}\|_2^2\leq \rho^{-2j} (\dmax^2+\|\qnf_{j|t}\|_2^2)$ we have
	\begin{flalign}  \nonumber
		&\rho^{-2k}\cf_{k|t}- \rho^{-2t}\cf_{t|t} \\&\quad = \mu \sum_{j=t}^{k-1} \rho^{-2j} (\dmax^2+\|\qnf_{j|t}\|_2^2)\geq  \mu \sum_{j=t}^{k-1} \|\lt w_{j|t}\|^2_2.\hspace{-1cm}&\label{eq:cf_sum}
	\end{flalign}
	Moreover, decompose $\P=(\begin{smallmatrix} \P_{11} & \P_{12} \\ \P_{21} & \P_{22} \end{smallmatrix})$ with $\P_{11} \in\S^\nfilt$, $\P_{22}\in\S^{\nx+\nK}$. 
	Then, due to~\eqref{eq:stab_LMI2}, we have $\P_{22}\succ 0$ such that $g(\GKx)
	=\filtx^\top \P_{11}\filtx+2\filtx^\top \P_{12} \GKx + \GKx^\top \P_{22} \GKx$ is a convex function.
	Hence, for $\initopt_t\in[0,1]$, we have $g\big(\initopt_t \eGKxnf_{t|t-1} + (1-\initopt_t)\eGKxe_t\big )\leq \initopt_t g(\eGKxnf_{t|t-1})+(1-\initopt_t) g(\eGKxe_t)$.
	With $\filtx_{t|t}=\filtx_{t|t-1}=\filtx_t$ and $\overline{\eGKxe}_t\triangleq \rho^{-t}\eGKxe_t$, we conclude
	\begin{align}\nonumber
		&\symbscalar{\P \chi_{t|t}}  = [\star]^\top\!\P\begin{bmatrix}\filtx_t;\,  \initopt_t \overline{\eGKxnf}_{t|t-1} + (1-\initopt_t) \overline{\eGKxe}_t \end{bmatrix}\\
		& \leq \initopt_t \chi_{t|t-1}^\top \P \chi_{t|t-1}+(1-\initopt_t)\ubar\chi_t^\top\P\ubar\chi_t. \label{eq:Pchi_conv}
	\end{align}
	An intermediate result, which we prove by induction, is that
	\begin{align}\label{eq:Pchi_bound}
		\chi_{t|t}^\top \P \chi_{t|t} + \sum_{j=0}^{t-1} \filty_j^\top \M \filty_j \leq \rho^{-2t}\cf_{t|t}
	\end{align}
	holds for all $t\geq 0$.
	The base case $t=0$ follows from
	\begin{align*}
		\chi_{0|0}^\top \P \chi_{0|0} & \refeq{\scriptsize\eqref{eq:Pchi_conv}}\leq \initopt_0 \symbscalar{\P \chi_{0|-1}} + (1-\initopt_0) \symbscalar{\P \ubar\chi_0} \\&\refeq{\scriptsize(\ref{eq:eGKxe_bound})}\leq \initopt_0 \cf_{0|-1} + (1-\initopt_0) \ce_0 = \cf_{0|0}.
	\end{align*}
	For the induction step from $t$ to $t+1$ we make use of $\filty_{t|t}=\filty_{t}$ and follow similar arguments
	\begin{align*}
		&\symbscalar \P \chi_{t+1|t+1}
		 \refeq{\scriptsize\eqref{eq:Pchi_conv}}\leq \initopt_{t+1} \symbscalar\P \chi_{t+1|t} + (1-\initopt_{t+1}) \symbscalar\P \ubar\chi_{t+1}                                                                                                     \\
		 & \refeq{\scriptsize\eqref{eq:delta1}}\leq  \initopt_{t+1} \big(\symbscalar{\P \chi_{t|t}} - \filty_{t}^\top \M\filty_{t} +\mu\|\lt\w_{t|t}\|_2^2\big) \\
		  &\qquad + (1-\initopt_{t+1}) \symbscalar{\P \ubar\chi_{t+1}}                              \\
		 & \refeq{(\ref{eq:eGKxe_bound},\ref{eq:Pchi_bound})}\leq \initopt_{t+1} \rho^{-2(t+1)}\left( \rho^{2} \cf_{t|t} + \mu\rho^2  \big(\|\qnf_{t|t}\|_2^2 +\dmax^2\big) \right)   \\
		 &\qquad +(1-\initopt_{t+1})\rho^{-2(t+1)}\ce_{t+1} -\sigma(t) \\
		 & \refeq{\eqref{eq:cf}}= \rho^{-2(t+1)} \left(\initopt_{t+1} \cf_{t+1|t} + (1-\initopt_{t+1}) \ce_{t+1} \right)-\sigma(t) \\
		 &=\rho^{-2(t+1)} \cf_{t+1|t+1} -\sigma(t)
	\end{align*}
	where we used the shorthand notation $\sigma(t)\triangleq\sum_{j=0}^{t} \filty_j^\top \M\filty_j$.
	Hence, we have established~\eqref{eq:Pchi_bound}.
	Next, we make use of $\chi_{k|t}^\top \ubar \X_i \chi_{k|t} = \filtx_{k|t}^\top \X_i \filtx_{k|t}$, the telescoping sum argument
    \begin{align}
        \sum_{j=t}^{k-1}\big(\symbscalar \P\chi_{j+1|t}-\chi_{j|t}^\top \P\chi_{j|t}\big)\! = \chi_{k|t}^\top \P\chi_{k|t} - \chi_{t|t}^\top \P\chi_{t|t}, \nonumber \\[-0.4cm] \label{eq:tel_sum}
    \end{align} 
	and $\filty_{j|t}=\filty_j$ for $j\in\I_{[0,t-1]}$ to compute
	\begin{align*}
		&0  \geq \sum_{j=t}^{k-1} \delta_1(j) + \delta_2(k) = - \mu \sum_{j=t}^{k-1} \|\lt \w_{j|t}\|^2_2 +  {\textstyle\frac{\alpha}{\gamma}}\|\overline \eznf_{k|t} \|_2^2 \\
		 &\hspace{3.6cm}-\beta \|\lt \w_{k|t}\|^2_2 - \chi_{t|t}^\top\P\chi_{t|t} + 
		 \\&\resizebox{\linewidth}{!}{$\underbrace{\sum_{j=t}^{k-1} \filty_{j|t}^\top \M_1 \filty_{j|t}+ \chi_{k|t}^\top \ubar \X_1\chi_{k|t} + \sum_{j=t}^k\filty_{j|t}^\top \M_2 \filty_{j|t} + \symbscalar \ubar\X_2 \chi_{k+1|t}}_{\geq -\sum_{j=0}^{t-1} \filty_j^\top \M\filty_j \text{ due to \eqref{eq:iqc}, $\M=\M_1+\M_2$, and $(\M_1,\X_1),(\M_2,\X_2)\in\MXset$}}$.}
	\end{align*}
	After adding~\eqref{eq:cf_sum} and~\eqref{eq:Pchi_bound} to this inequality, we obtain
	\begin{align*}
		0 & \geq -\beta \|\lt \w_{k|t}\|^2_2 - \rho^{-2k} \cf_{k|t} +  {\textstyle\frac{\alpha}{\gamma}}\|\overline \eznf_{k|t} \|_2^2.
	\end{align*}
	With $\|\lt\w_{k|t}\|^2_2\leq \rho^{-2k}\|\qnf_{k|t}\|_2^2+\rho^{-2k}\dmax^2$, $\overline \eznf_{k|t}=\rho^{-k}\eznf_{k|t}$, and after multiplication by $\rho^{2k}\frac{\gamma}{\alpha}$, we obtain~\eqref{eq:constraint_tightening}.
\end{proof}

\begin{proof}[Proof of Theorem~\ref{thm:iss}]
	We show recursive feasibility by constructing the feasible candidate solution $\initopt^\mathrm{fc}_t\triangleq 1$ and
	\begin{align}\label{eq:unf_ws}
		\u^\mathrm{fc}_{k|t} = \begin{cases}
			\unf_{k|t-1} & \text{for } k \in \I_{[t,t+N-2]} \\
			0 & \text{for } k=t+N-1.
		\end{cases}
	\end{align}
	Due to $\initopt_t^\mathrm{fc}=1$, the resulting trajectories start at $\GKxnf^\mathrm{fc}_{t|t} = \GKxnf_{t|t-1}$ and  $\cf_{t|t}^\mathrm{fc}=\cf_{t|t-1}$.
	Thus and due to~\eqref{eq:unf_ws}, we have
	\begin{align}\nonumber
		\GKxnf^\mathrm{fc}_{k+1|t} &= \GKxnf_{k+1|t-1}, & 
		\cf^\mathrm{fc}_{k+1|t} &= \cf_{k+1|t-1}, \\
		\qnf^\mathrm{fc}_{k|t} &= \qnf_{k|t-1}, &
		\znf^\mathrm{fc}_{k|t} &= \znf_{k|t-1} \label{eq:znf_ws}
	\end{align}
	for $k \in \I_{[t,t+N-1]}$ and for $k=t+N$ we have
	\begin{align*}
		\GKxnf^\mathrm{fc}_{t+N|t} &= \A\GK \GKxnf_{t+N-1|t-1}, \\
		\cf^\mathrm{fc}_{t+N|t} &= \rho^2 \cf_{t+N-1|t-1}+\mu\rho^2\big(\|\qnf^\mathrm{fc}_{t+N-1|t}\|_2^2+\dmax^2\big), \\
		\qnf^\mathrm{fc}_{t+N-1|t} &= \C\GK\q \GKxnf_{t+N-1|t-1}, \quad 
		\znf^\mathrm{fc}_{t+N-1|t} = \C\GK\z\GKxnf_{t+N-1|t-1}.
	\end{align*}
	Assume feasibility at time $t-1$, then $\cf_{t+N-1|t-1}\leq \cf^\mathrm{f}$ and
	$\|\tswroot \GKxnf_{t-1+N|t-1} \|_2^2  \leq \GKxnf^\mathrm{f}$.
	Hence, we have
	\begin{align*}
		\|\tswroot \GKxnf^\mathrm{fc}_{t+N|t} \|_2^2&= \|\tswroot \A\GK \GKxnf_{t+N-1|t-1}\|_2^2 \\
		&\refeq{\eqref{eq:ts_inv1}}\leq \|\tswroot \GKxnf_{t+N-1|t-1}\|_2^2 \leq \GKxnf^\mathrm{f}.
	\end{align*}
	Due to
	\begin{align}\label{eq:ts_qbound}
		\|\qnf_{t+N-1|t}^\mathrm{fc} \|_2^2
		 \leq \max_{\|\tswroot\GKxnf\|_2^2 \leq \GKxnf^{\mathrm{f}}} \|\C\GK\q \GKxnf \|_2^2 = \|\C\GK\q \tswroot^{-1}\|_2^2 \GKxnf^\mathrm{f}
	\end{align}
	we also have
	\begin{align}\label{eq:ts_cbound}
		\cf^\mathrm{fc}_{t+N|t} &\leq \rho^2 \cf^\mathrm{f}+\mu\rho^2\big(\|\C\GK\q \tswroot^{-1}\|_2^2 \GKxnf^\mathrm{f}+\dmax^2\big) \refeq{\eqref{eq:ts_inv2}}\leq \cf^\mathrm{f}
	\end{align}
	and thus the candidate indeed satisfies~\eqref{eq:mpc:term}. 
	Finally, we have 
	\begin{align*}
		&\znf^\mathrm{fc}_{t+N-1|t} \leq \|\znf^\mathrm{fc}_{t+N-1|t}\|_\infty \leq  \max_{\|\tswroot \GKxnf\|_2^2 \leq \GKxnf^\mathrm{f}} \|\C\GK\z \GKxnf \|_\infty \\
		&\quad=\| \C\GK\z \tswroot^{-1}\|_{2\shortrightarrow\infty} \sqrt{\GKxnf^\mathrm{f}}\\
		&\quad\refeq{\eqref{eq:ts_feas}}\leq 1 - \sqrt{\frac{\gamma}{\alpha} \cf^\mathrm{f}+\frac{\gamma}{\alpha}\beta (\|\C\GK\q\tswroot^{-1}\|_2^2\GKxnf^\mathrm{f} + \dmax^2)}\\
		&\quad\refeqleft{(\ref{eq:ts_qbound})}\leq 1 - \sqrt{\frac{\gamma}{\alpha} \cf_{t+N-1|t}^\mathrm{fc}+\frac{\gamma}{\alpha}\beta(\|\qnf_{t+N-1|t}^\mathrm{fc} \|_2^2 + \dmax^2)}
	\end{align*}
	as $\cf_{t+N-1|t}^\mathrm{fc}=\cf_{t+N-1|t-1} \leq \cf^\mathrm{f}$.
	This implies that the candidate indeed satisfies~\eqref{eq:mpc:constraint} at $k=t+N-1$.
	For $k\in\I_{[t,t+N-2]}$ the constraint~\eqref{eq:mpc:constraint} holds due to~\eqref{eq:znf_ws} and feasibility at time $t-1$.
	Hence, the MPC scheme is feasible at time $t$ if it is feasible at time $t-1$, i.e., it is recursively feasible.
	Further, constraint satisfaction $\z_t\leq 1$ follows for all $t\geq 0$ as $\z_t = \znf_{t|t}+\eznf_{t|t}$,~\eqref{eq:constraint_tightening}, and~\eqref{eq:mpc:constraint} with $k=t$ hold for all $t\geq 0$.

	Finally, we show input-to-state stability.
	Due to Assumptions~\ref{ass:iqc},~\ref{ass:p2p_LMIs} we can apply~\cite[Theorem~3]{Schwenkel2025} to conclude that $\GK\star\Delta$ is $\ell_{2,\rho}$-stable.
	Hence, the dynamics~\eqref{eq:GKxnp_sys} and~\eqref{eq:eGKxnp_sys} of $\GKxnp$ and $\eGKxnp$ are $\ell_{2,\rho}$-stable and thus constants $\tilde \alpha_1$, $\tilde\alpha_2$ exist with 
	\begin{align*}
		\sum_{k=0}^{t}\rho^{-2k} \|\eGKxnp_k\|_2^2 &\leq \tilde\alpha_1 \|\exnf_0\|_2^2 +\tilde  \alpha_2\sum_{k=0}^{t-1}\rho^{-2k} (\|\d_k\|_2^2+\|\qnp_k\|_2^2) \\
		\sum_{k=0}^{t}\rho^{-2k} \|\GKxnp_k\|_2^2 &\leq \tilde\alpha_1 \|\xnf_0\|_2^2 + \tilde\alpha_2 \sum_{k=0}^{t-1}\rho^{-2k} \|\unf_k\|_2^2
	\end{align*}
	holds for all $t\geq 0$.
	As a consequence and due to $\xe_0 = \x_0-\exe_0$, $\|\GKx_k\|_2^2=\|\GKxnp_k+\eGKxnp_k\|_2^2\leq 2\|\GKxnp_k\|_2^2+2\|\eGKxnp_k\|_2^2$, and $\qnp_k = \C\GK\q \GKxnp_k + \D\GK\q\u \unf_k$, we infer that there exist constants $\alpha_2\geq 0$ and $\alpha_3\geq 0$ such that for all $t\geq 0$ we have
	\begin{align}\nonumber
		\|\GKx_t\|_2^2 &\le \sum_{k=0}^{t}\rho^{2t-2k} \|\GKx_k\|_2^2\le \rho^{2t} \alpha_2 (\|\x_0\|_2^2+\|\exe_0\|_2^2) \\[-0.2cm]&\qquad\qquad\  + \alpha_3 \sum_{k=0}^{t-1}\rho^{2t-2k} (\|\d_k\|_2^2+\|\unf_k\|_2^2)\label{eq:iss_half_way}
	\end{align}
	What remains to verify~\eqref{eq:iss} is a bound on $\unf_k$.
	To this end, we use optimality of $J(\GKxnf_{t+1|t+1},\unf_{\cdot|t+1})$ to conclude 
	\begin{align*}
		&J(\GKxnf_{t+1|t+1},\unf_{\cdot|t+1}) \leq J(\GKxnf_{t+1|t},\u^\mathrm{fc}_{\cdot|t+1})\\
		& = J(\GKxnf_{t|t},\unf_{\cdot|{t}}) - \left\| \begin{bmatrix} \GKxnf_{t|t};\, \unf_{t|t} \end{bmatrix} \right\|_\scw^2 \\&\qquad   \underbrace{-\|\GKxnf_{t+N|t}\|_\tcw^2+\left\| \begin{bmatrix} \GKxnf_{t+N|t}^\mathrm{fc} ;\, 0 \end{bmatrix} \right\|_\scw^2+\|\GKxnf_{t+1+N|t+1}^\mathrm{fc}\|_\tcw^2}_{\leq 0 \text{ due to \eqref{eq:ass_tcw}}}.
	\end{align*}
	Further, standard arguments (cf. proof of~\cite[Theorem~5]{Schwenkel2023c}) yield $J(\GKxnf_{t|t},\unf_{\cdot|t})\leq \alpha_6(\|\GKxnf_{t|t}\|_2)$ for some $\alpha_6\in\mathcal{K}_\infty$.
	As $\scw\succ0$, there is $\alpha_7\in\Kinf$ such that we have $\alpha_6(\|\GKxnf_{t|t}\|_2) \leq \alpha_7 \Big( \left\| \begin{bmatrix} \GKxnf_{t|t};\, \unf_{t|t} \end{bmatrix} \right\|_\scw \Big) $.
	Hence, 
	\begin{align*}
		J(\GKxnf_{t+1|t+1},\unf_{\cdot|t+1})\!-\!J(\GKxnf_{t|t},\unf_{\cdot|t})\!\leq\!- \alpha_7^{-1} \big( J(\GKxnf_{t|t},\unf_{\cdot|t})\big)
	\end{align*}
	and thus following standard Lyapunov arguments, there exists a function $\beta_2\in\KL$ such that $\|\GKxnf_{t|t}\|_2 \leq \beta_2(\|\GKxnf_{0|0}\|_2,t)$.
	Further, due to $\scw\succ0$ and $\tcw\succeq0$ there exists $\alpha_8\in\Kinf$ such that $\|\unf_{t|t}\|_2^2 \leq \alpha_8\big(J(\GKxnf_{t|t},\unf_{\cdot|t})\big) \leq\alpha_8\big(\alpha_6(\|\GKxnf_{t|t}\|_2)\big) \leq \alpha_8\big(\alpha_6(\beta_2(\|\GKxnf_{0|0}\|_2,t))\big)\triangleq \varepsilon_1(t)$, where we introduced $\varepsilon_1$ for the ease of notation. Thus, we have
	\begin{align} \label{eq:iss_mid2}
		\sum_{k=0}^{t-1} \rho^{2t-2k} \|\unf_k\|_2^2 &\leq \sum_{k=0}^{t-1} \rho^{2t-2k} \varepsilon_1(t) \triangleq \varepsilon_2(t).
	\end{align}
	By definition of $\varepsilon_2$, we have $\varepsilon_2(t+1) = \rho^2 \varepsilon_2(t) + \varepsilon_1(t)$. 
	Further, let $\rho<\rho_1<1$ and define $\lt \varepsilon_2(t) \triangleq \rho_1^{-2t} \varepsilon_2(t)$ and $\lt \varepsilon_1(t) \triangleq \rho_1^{-2t} \varepsilon_1(t)$, then $\lt \varepsilon_2(t+1) = (\frac{\rho}{\rho_1})^2 \lt \varepsilon_2(t) + \rho_1^{-2} \lt\varepsilon_1(t)$. 
	Next, we show $\lt\varepsilon_2(t) \leq \frac{1}{\rho_1^2-\rho^2} \max_{k\in\I_{[0,t-1]}} \lt\varepsilon_1(k)$ by induction.
	The base case $t=1$ follows from $\lt\varepsilon_2(1)=\rho_1^{-2} \lt\varepsilon_1(0) \leq \frac{1}{\rho^2_1-\rho^2} \lt\varepsilon_1(0)$.
	For the induction step $t\to t+1$ we have 
	\begin{align*}
		\lt \varepsilon_2(t+1) &= ({\textstyle\frac{\rho}{\rho_1}})^2 \lt \varepsilon_2(t) + \rho_1^{-2} \lt\varepsilon_1(t) \\
		&\leq {\textstyle(\frac{\rho}{\rho_1})^2 \frac{1}{\rho_1^2-\rho^2}} \max_{k\in\I_{[0,t-1]}} \lt\varepsilon_1(k)+ \rho_1^{-2} \lt\varepsilon_1(t)\\
		&\leq \underbrace{\left({\textstyle(\frac{\rho}{\rho_1})^2 \frac{1}{\rho_1^2-\rho^2}}+\rho_1^{-2}\right)}_{=\frac{1}{\rho_1^2-\rho^2}} \max_{k\in\I_{[0,t]}} \lt\varepsilon_1(k).
	\end{align*}
	Hence, we have $\varepsilon_2(t) \leq \frac{1}{\rho_1^2-\rho^2} \max_{k\in\I_{[0,t-1]}} \rho^{2t-2k}_1 \varepsilon_1(k)$ and plugging this bound into~\eqref{eq:iss_mid2}, yields 
	\begin{flalign}\nonumber
		&\sum_{k=0}^{t-1} \rho^{2t-2k} \|\unf_k\|_2^2 \leq \beta_3(\|\GKxnf_{0|0}\|_2,t)\\[-0.15cm]
		&\ \ \quad \triangleq {\textstyle\frac{1}{\rho_1^2-\rho^2}} \max_{k\in\I_{[0,t-1]}} \rho^{2t-2k}_1 \alpha_8\big(\alpha_6(\beta_2(\|\GKxnf_{0|0}\|_2,k))\big).\!\!\!\!\!&\label{eq:beta3}
	\end{flalign}
	Since $\beta_2(\cdot,k)\in\Kinf$ for all $k\in\N$, we have $\beta_3(\cdot,t)\in\Kinf$ for all $t\in\N$. 
	Further, $\beta_3(\|\GKxnf_{0|0}\|_2,t)$ is monotonically decreasing in $t$ and going to $0$ as $t\to\infty$.
	Hence, $\beta_3\in\KL$.
	Finally, note that $\|\GKxnf_{0|0}\|_2 = \|\xe_0\|_2\leq \|x_0\|_2+\|\exe_0\|_2$ and plug~\eqref{eq:beta3} into~\eqref{eq:iss_half_way} to obtain~\eqref{eq:iss} with $\alpha_1\in\Kinf$ defined by $\alpha_1(\|\d\|_\peak) \triangleq \alpha_3 \frac{1}{1-\rho^2} \|\d\|_\peak^2$ and $\beta_1\in\KL$ defined by $\beta_1\left(a,t\right) \triangleq \rho^{2t} \alpha_2 (a) + \alpha_3 \beta_3(a,t)$.
\end{proof}

\begin{proof}[Proof of Theorem~\ref{thm:obs_p2p}]
	Due to Assumption~\ref{ass:no_alg_loop} the estimation error $\eGKxe_t$ and hence $\zobs_t$ is independent of $\unf_{t}$.
	Hence, to bound $\zobs_t$, we take the $\unf_t$ that provides the smallest bound.
	In particular, we minimize
	\begin{align}\nonumber
		\min_{\unf_{t}} \|\qnp_t\|_2^2 &= \min_{\unf_{t}} \|\C\GK\q\GKxnp_t + \D\GK\q\u \unf_{t}\|_2^2 \\&= \min_{b_t} \|\C\GK\q\GKxnp_t + \mathscr{D} b_t\|_2^2 = \|\mathscr{C}\GKxnp_t\|_{2}^2, \label{eq:qstar}
	\end{align}
	which holds as the minimizer is $b_t^\star =-(\mathscr{D}^\top \mathscr{D})^{-1}\mathscr{D}^\top \C\GK\q \GKxnp_t$.
	Let ${\unf_t}^{\star}$ satisfy $\D\GK\q\u {\unf_t}^{\star}= \mathscr{D} b_t^\star$.
	Let $\p^\star_t$, $\q^\star_t$, $\qnp^\star_t$, $\xi^\star_{t+1}$, $\filty_t^\star$, and ${\zobs_t}^\star$ be the corresponding values that follow if we use $\unf_t = {\unf_t}^{\star}$ in~\eqref{eq:DGKx_sys},~\eqref{eq:GKxnp_sys},~\eqref{eq:xi_sys}.
	Due to Assumption~\ref{ass:no_alg_loop}, we know that $\D\GK\y\p \p^\star_t = \D\GK\y\p \p_t$ and thus $\zobs_t={\zobs_t}^\star$.
	Further, let ${\barwobs_t}^\star = \rho^{-t}[\qnp^\star_t;\, \d_t]$.
	Multiplying~\eqref{eq:obs_LMI2} from the right and left by $\begin{bmatrix}\xi_t;\, \lt\p^{\star}_t;\, {\barwobs_t}^{\star}\end{bmatrix}$ and its transpose yields
	\begin{align*}
		\xi_t^\top (\ubar \X_3-\Pobs) \xi_t+ \xi_{t+1}^{\star\top} \ubar \X_4 \xi_{t+1}^\star+\filty^{\star\top}_t\M_4\filty^\star_t +\symbscalar{\P \barzobs_t}&\\-\beta^\mathrm{o} \|{\barwobs_t}^\star\|^2_2&\leq 0.
	\end{align*}
	Since $(\M_4,\X_4)\in\MXset$ and due to the IQC~\eqref{eq:iqc} we have
	\begin{align*}
		\sum_{j=0}^{t-1} \filty_j^\top \M_4 \filty_j +\filty^{\star\top}_t\M_4\filty^\star_t + \xi_{t+1}^{\star\top} \ubar \X_4 \xi_{t+1}^\star \geq 0
	\end{align*}
	and hence $\delta_4(t)\leq 0$ with $
		\delta_4(t) \triangleq \xi_t^\top\!(\ubar \X_3\!-\!\Pobs) \xi_t\! -\!\! \sum_{j=0}^{t-1}\! \filty_j^\top \M_4\filty_j\!+\!\symbscalar{\!\P \barzobs_t}\!-\!\beta^\mathrm{o} \|{\barwobs_t}^\star\|^2_2$.
	Multiplying~\eqref{eq:obs_LMI1} from right and left by $\begin{bmatrix}\xi_j;\, \lt\p_j;\, {\barwobs_j}\end{bmatrix}$ and its transpose yields
	\begin{align*}
		\delta_3(j) & \triangleq\xi_{j+1}^\top \Pobs\xi_{j+1}-\xi_j^\top \Pobs\xi_j+\filty_j^\top\Mobs\filty_j - \muobs \|\barwobs_j\|^2_2 \leq 0.
	\end{align*}
	Due to~\eqref{eq:cp}, $\cp_0\geq \xi_0^\top \Pobs\xi_0$, and $\barwobs_{j}=\rho^{-j} \wobs_{j}$ we have
	\begin{align}\nonumber
		\cp_{t} & = \muobs \sum_{j=0}^{t-1} \rho^{2(t-j)} \underbrace{(\dmax^2+\|\qnp_j\|_2^2)}_{\geq \|\wobs_j\|_2^2} + \rho^{2t}\cp_0\\
		& \geq \rho^{2t} \xi_0^\top \Pobs \xi_0 +\rho^{2t}\muobs {\sum}_{j=0}^{t-1} \|\barwobs_j\|^2_2.\label{eq:cp_sum}
	\end{align}
	Finally, we use $\beta^\mathrm{o}\geq 0$, $(\M_3,\X_3)\in\MXset$,~\eqref{eq:qstar}, and the analogue to the telescoping sum argument from~\eqref{eq:tel_sum} to obtain
	\begin{align}\nonumber
		&0  \geq \sum_{j=0}^{t-1} \delta_3(j) + \delta_4(t) =   \symbscalar{\P \barzobs_t} +\sum_{j=0}^{t-1} \filty_j^\top \M \filty_j - \beta^\mathrm{o} \|{\barwobs_t}^\star\|^2_2         \\ \nonumber
		  & \qquad  \underbrace{- \xi_0^\top \Pobs \xi_0  - \muobs \sum_{j=0}^{t-1} \|\barwobs_j\|^2_2}_{\geq -\rho^{-2t} \cp_t \text{ due to \eqref{eq:cp_sum}}}
		+  \underbrace{\sum_{j=0}^{t-1} \filty_j^\top \M_3 \filty_j+ \xi_t^\top \ubar \X_3\xi_t}_{\geq 0 \text{ due to \eqref{eq:iqc}}}                                                             \\ \nonumber
		  & \geq -\rho^{-2t} \cp_{t}\! -\! \rho^{-2t} \beta^\mathrm{o} (\dmax^2\!+\!\|\mathscr{C} \GKxnp_t\|^2_2)\!+\!  \symbscalar{\!\P \barzobs_t}\!+\!\!\sum_{j=0}^{t-1} \filty_j^\top \M \filty_j.
	\end{align}
	Multiplying this inequality by $\rho^{2t}$ and using $\rho^{2t}\symbscalar{\P \barzobs_t}=\symbscalar{\P \zobs_t}$ proves the bound~\eqref{eq:eGKxe_bound}.
	Further, due to $\|\mathscr{C} \GKxnp_t\|^2_2\leq \|\qnp\|_\peak^2$ and~\eqref{eq:ce}, we have $\ce_t\leq \cp_t + \beta^\mathrm{o}(\|\qnp\|_\peak^2+\dmax^2)$.
	Moreover, if we have $\cp_0\leq \alpha\muobs(\|\qnp\|_\peak^2+\dmax^2)$, then we can show by induction that also $\cp_t\leq \alpha{\muobs}(\|\qnp\|_\peak^2+\dmax^2)$ holds.
	The base case is given by assumption, the induction step follows as $\cp_{t+1} \leq   \rho^{2}\cp_t + \rho^{2} \muobs (\|\qnp\|_\peak^2+\dmax^2) \leq \rho^2\left( \alpha + 1 \right)\muobs(\|\qnp\|_\peak^2+\dmax^2)=\alpha\muobs (\|\qnp\|_\peak^2+\dmax^2)$.
	Hence, we have $\ce_t\leq \cp_t + \beta^\mathrm{o}(\|\qnp\|_\peak^2+\dmax^2) \leq  \alpha\muobs (\|\qnp\|_\peak^2+\dmax^2) + \alpha(\gamobs-\muobs)(\|\qnp\|_\peak^2+\dmax^2) =\alpha\gamobs(\|\qnp\|_\peak^2+\dmax^2) $, which establishes~\eqref{eq:p2p_obs}.
\end{proof}

\arxivonly{
\begin{proof}[Proof of Theorem~\ref{thm:obs_syn}]
	Let us first show, how to construct a solution for~\eqref{eq:obs_LMI1_syn},~\eqref{eq:obs_LMI2_syn} from a solution of~\eqref{eq:obs_LMI1},~\eqref{eq:obs_LMI2} with $\Pobs_{22}\succ 0$.
	We assume without loss of generality that $\Pobs_{21}$ is invertible.
	If $\Pobs_{21}$ were singular, then we slightly perturb it to make it invertible while still satisfying the strict inequalities~\eqref{eq:obs_LMI1},~\eqref{eq:obs_LMI2}.
	Next, we apply a similar variable transformation as in~\cite{Scherer2008}.
	Define $\Psynobs_{1} \triangleq \Pobs_{11} - \Pobs_{12}(\Pobs_{22})^{-1} \Pobs_{21}$, $\Psynobs_{2}\triangleq\Pobs_{11}$, $\Zsynobs \triangleq -(\Pobs_{22})^{-1} \Pobs_{21}$, and $\Ysyntrafoobs = \begin{pmatrix} I & I \\ \Zsynobs & 0 \end{pmatrix}$.
	Then
	\begin{align} \label{eq:Psynobs}
		{\Ysyntrafoobs}^\top \Pobs \Ysyntrafoobs &= \begin{pmatrix} \Psynobs_{1}& \Psynobs_{1}\\ \Psynobs_{1} & \Psynobs_{2}\end{pmatrix},\\
		\symb{\Pobs \begin{pmatrix} I & 0\\ 0 & (\Pobs_{12})^{-1}\end{pmatrix}} &= \begin{pmatrix} \Psynobs_{2}&  I \\ I & (\Psynobs_{2}-\Psynobs_{1})^{-1}\end{pmatrix}.
	\end{align}
	Furthermore, define the new estimator variables as $\Ksynobs \triangleq \Pobs_{12}\A\L\Zsynobs$, $\Lsynobs \triangleq \Pobs_{12}\BB\L$, $\Msynobs\triangleq\CC\L\Zsynobs$, $\Nsynobs\triangleq\DD\L$.
	Given these new variables and~\eqref{eq:ABCDgpobs},~\eqref{eq:ABCDobs}, we observe
	\begin{align}\nonumber                          &
             \blkmat{cc:c}{
             I                                  & 0                                  & 0                                                           \\
             0                                  & \Pobs_{12}                         & 0 \newblkdash 0              & 0                          & I }
             \blkmat{c:cc}{
             \A\obsrho                          & \B\obsrho\p                        & \B\obsrho\w      \newblkdash
             \C\obs \filty                      & \D\obs \filty\p                    & \D\obs\filty\w                                              \\
             \C\obs \z                          & \D\obs\z\p                         & \D\obs\z\w
             }
             \blkmat{c:c}{
             \Ysyntrafoobs                      & 0 \newblkdash
             0                                  & I }                                                                                              \\ &
             \qquad =
             \blkmat{cc:cc}{
             \A\gpobsrho                        & \A\gpobsrho                        & \B\gpobsrho\p            & \B\gpobsrho\w                    \\
             \Asynobs_1                         & \Asynobs_2                         & \Bsynobs_1               & \Bsynobs_2         \newblkdash
             \C\gpobs\filty                     & \C\gpobs\filty                     & \D\gpobs\filty\p         & \D\gpobs\filty\w                 \\
             \begin{pmatrix}I & 0 \end{pmatrix} & \begin{pmatrix}I & 0 \end{pmatrix} & 0                        & 0                   \\
             \Csynobs_1                         & \Csynobs_2                         & \Dsynobs_{1}             & \Dsynobs_2
             }.\label{eq:synsysobs}
	\end{align}
	Using~\eqref{eq:Psynobs},~\eqref{eq:synsysobs} and multiplying~\eqref{eq:obs_LMI1} from right by $\diag(\Ysyntrafoobs,I)$ and from left by its transpose, we obtain
			$\mathcal{O}_1 + \symbscalar{(\Psynobs_{2}-\Psynobs_{1})^{-1} \begin{pmatrix}\Asynobs_1     & \Asynobs_2  & \Bsynobs_1       & \Bsynobs_2	\end{pmatrix}}\prec 0$.
	We know that $\Psynobs_{2}-\Psynobs_{1} = \Pobs_{12} (\Pobs_{22})^{-1} \Pobs_{21} \succ 0$ as $\Pobs_{22}\succ 0$ and $\Pobs_{21}$ is invertible.
	Hence, we can apply a Schur complement to arrive at~\eqref{eq:obs_LMI1_syn}.
	Moreover, we use\footnote{In a slight abuse of notation we use $\ubar\X_i = \diag( \X_i, 0_{n\times n})$ with the dimension $n=\nx+\nK$ and $n=\nx+\nK+\nL$ simultaneously.} ${\Ysyntrafoobs}^\top (\Pobs - \ubar\X_3) \Ysyntrafoobs = \begin{pmatrix}\Psynobs_1-\ubar\X_3 & \Psynobs_1-\ubar\X_3\\\Psynobs_1-\ubar\X_3 & \Psynobs_{2}\end{pmatrix}$ and multiply~\eqref{eq:obs_LMI2} from right by $\diag(\Ysyntrafoobs,I)$ and from left by its transpose to obtain

$		\mathcal{O}_2 + \symbscalar{\P_{22} \begin{pmatrix}\Csynobs_1 & \Csynobs_2 & \Dsynobs_{1} & \Dsynobs_2 \end{pmatrix}}\prec 0.$
	Since $\P_{22}\succ 0$ due to~\eqref{eq:stab_LMI2}, we can apply a Schur complement to arrive at~\eqref{eq:obs_LMI2_syn}.

	To show the other direction, we set $\Pobs_{21}=I$ and reverse the transformation steps to see that~\eqref{eq:Pobs_syn},~\eqref{eq:L_syn} is a solution of~\eqref{eq:obs_LMI1},~\eqref{eq:obs_LMI2}.
	Further, note that $\Pobs_{22} = (\Psynobs_2-\Psynobs_1)^{-1}\succ 0$ due to the lower right block of~\eqref{eq:obs_LMI1_syn}.
\end{proof}
} 
\end{document}